\documentclass[draftclsnofoot, onecolumn, 12pt]{IEEEtran}

\usepackage{mathtools}

\usepackage{floatflt}

\usepackage{color}
\usepackage{cite}

\usepackage{graphicx}
\usepackage{amsmath}
\usepackage{amsthm}
\usepackage{amssymb}

\usepackage{verbatim}

\usepackage{psfrag,array}

\usepackage{subfigure}

\usepackage{stfloats}

\usepackage{todonotes}

\usepackage{amsmath}
\usepackage{epstopdf}
\usepackage{algorithmic}

\newcommand{\p}[1]{\mathop{\mbox{\it p} } }

\renewcommand{\vec}[1]{\ensuremath{\boldsymbol{#1}}}
\newcommand{\be}{\begin{equation}}
\newcommand{\ee}{\end{equation}}
\newcommand{\ba}{\begin{array}}
\newcommand{\ea}{\end{array}}
\newcommand{\bea}{\begin{eqnarray}}
\newcommand{\eea}{\end{eqnarray}}
\newcommand{\bean}{\begin{eqnarray*}}
\newcommand{\eean}{\end{eqnarray*}}

\newcommand{\rmh}{^{\rm H}}
\newcommand{\rmt}{^{\rm T}}

\renewcommand{\Re}{\mathcal{R}}

\definecolor{white}{rgb}{1,1,1}

\newtheorem{theorem}{Theorem}

\newtheorem{property}{Property}

\newtheorem{corollary}{Corollary}
\newtheorem{remark}{Remark}

\begin{document}

\title{Beyond Massive-MIMO: The Potential of Positioning with Large Intelligent Surfaces}
\author
{
Sha Hu, Fredrik Rusek, and Ove Edfors \\
Department of Electrical and Information Technology,\\
Lund University, Lund, Sweden\\
$\big($\{firstname.lastname\}@eit.lth.se$\big)$.
\thanks{This paper will be presented in part \cite{HRE172} in IEEE 86th Vehicular Technology Conference (VTC-Fall), Toronto, Canada, 24-27 Sep. 2017.}
}
\maketitle

\begin{abstract}
We consider the potential for positioning with a system where antenna arrays are deployed as a large intelligent surface (LIS), which is a newly proposed concept beyond massive-MIMO where future man-made structures are electronically active with integrated electronics and wireless communication making the entire environment \lq\lq{}intelligent\rq\rq{}. In a first step, we derive Fisher-information and Cram\'{e}r-Rao lower bounds (CRLBs) in closed-form for positioning a terminal located perpendicular to the center of the LIS, whose location we refer to as being on the central perpendicular line (CPL) of the LIS. For a terminal that is not on the CPL, closed-form expressions of the Fisher-information and CRLB seem out of reach, and we alternatively find approximations of them which are shown to be accurate. Under mild conditions, we show that the CRLB for all three Cartesian dimensions ($x$, $y$ and $z$) decreases quadratically in the surface-area of the LIS, except for a terminal exactly on the CPL where the CRLB for the $z$-dimension (distance from the LIS) decreases linearly in the same. In a second step, we analyze the CRLB for positioning when there is an unknown phase $\varphi$ presented in the analog circuits of the LIS. We then show that the CRLBs are dramatically increased for all three dimensions but decrease in the third-order of the surface-area. Moreover, with an infinitely large LIS the CRLB for the $z$-dimension with an unknown $\varphi$ is 6 dB higher than the case without phase uncertainty, and the CRLB for estimating $\varphi$ converges to a constant that is independent of the wavelength $\lambda$. At last, we extensively discuss the impact of centralized and distributed deployments of LIS, and show that a distributed deployment of LIS can enlarge the coverage for terminal-positioning and improve the overall positioning performance. 
\end{abstract}


\section{Introduction}
Wireless communication has evolved from few and geographically distant base stations to more recent concepts involving a high density of access points, possibly with many antenna elements on each. A Large Intelligent Surface (LIS) is a newly proposed concept in wireless communication that is envisioned in \cite{HRE171}, where future man-made structures are electronically active with integrated electronics and wireless communication making the entire environment \lq\lq{}intelligent\rq\rq{}. We foresee a practical implementation of LIS as a compact integration of a vast amount of tiny antenna-elements with reconfigurable processing networks. Antennas on the surface cooperate to transmit and sense signals, both for communication and other types of functionality. Machine learning \cite{ML17} can bring intelligence in the systems both for autonomous operation of the system and for new functionality. One such application is depicted in Fig. \ref{fig1}, where three different terminals are communicating to LIS in an outdoor and indoor scenarios, respectively. 

The LIS concept can be seen as an extension of earlier research in several other fields. One strong relation is to the massive-MIMO concept \cite{M10, MM12, MM14}, where large arrays comprising hundreds of antennas are used to achieve massive gains in spectral and energy efficiencies. However, LIS scales up beyond the traditional antenna array, and implies a clean break with the traditional access-point/base-station concept, as the entire environment is active in the communication. The natural limit of this evolution is that the LISs in an environment act as transmitting and receiving structures, which allows for an unprecedented focusing of energy in the three-dimensional space which enables, besides unprecedented data-rates, wireless charging and remote sensing with extreme precision. This makes it possible to fulfill the most grand visions in 5G communication \cite{AZ14} and Internet of Things \cite{IOT} systems for providing connections to billions of devices.

A concept somewhat similar to what we call LIS seems to be first mentioned in the eWallpaper project at UC Berkeley \cite{ewall}, and in \cite{HRE171} we carry out a first analysis on information-transfer capabilities of the LIS, and show that the number of signal-space dimensions per square-meter ($m^2$) deployed surface-area is $\pi/\lambda^2$, where $\lambda$ is the wavelength, and the capacity that can be harvested per $m^2$ surface-area is linear in the average transmit power, rather than logarithmic as in a traditional massive-MIMO deployment. Following \cite{HRE171}, in this paper we take a first look at the potential of using LIS for terminal-positioning, where we assume that a terminal to be positioned is equipped with a single-antenna and located in a three-dimensional space in front of the LIS. For analytical tractability, although we do not deal with more complicated geometries, our results are fundamental in the sense that positioning of objects in rich scattering environments \cite{JY07, WW15} can be done in two steps: i) estimating the positions of a number of reflecting objects in the environment with line-of-sight (LoS) propagation to the LIS, and ii) backward computation of the position of the object of interest. Thus, the results obtained in this work are instrumental for understanding the accuracy in the first step for positioning with scatters.

In this work, we first derive the Cram\'{e}r-Rao lower bounds (CRLBs) for positioning a terminal on the central perpendicular line (CPL) in closed-form, where the CPL is the line perpendicular to the LIS and crossing the LIS at its center point as shown in Fig. \ref{fig2}. For remaining cases, as closed-form expressions seem out of reach and we approximate the Fisher-information and CRLB in closed-form, which are shown to be accurate under mild conditions. We also show that, the CRLB in general decreases quadratically in the surface-area of the LIS, except for a terminal on the CPL where the CRLB for $z$-dimension decreases just linearly in the same. Meanwhile, the impact of wavelength is $\sim\!\lambda^2$. These scaling laws play in favor of a LIS when compared to other positioning technologies e.g., optical systems \cite{MT11}. A LIS can compensate for its, comparatively, large wavelength by a much larger aperture. 

Besides, we also analyze the CRLB for positioning when there is an unknown phase $\varphi$ presented in the analog circuits of the LIS, in which case the CRLBs for all dimensions are dramatically increased by $\varphi$ and in general decrease in the third-order of the surface-area. Therefore, LIS has significant gains over traditional massive MIMO for positioning as LIS has a much larger surface-area. Furthermore, the CRLB for estimating $\varphi$ is usually significantly large and about $\frac{4\pi^2}{\lambda^2}$ times of the CRLB for the $z$-dimension. Moreover, for an infinitely large LIS, the CRLB for the $z$-dimension with unknown $\varphi$ is 6 dB higher than with known $\varphi$, and the CRLB for estimating $\varphi$ converges to a constant\footnote{Note that, all CRLBs and their limits considered in this paper can be linearly scaled down by the signal-to-noise ratio (SNR) as a natural result.}. 

Then, we also extensively discuss the impact of deployments with a single centralized LIS and multiple distributed smaller LISs constrained to the same total surface-area. We show that, a distributed deployment with splitting the single LIS into 4 small LISs can extend the range of positioning and provide better average CRLB than the centralized deployment under the case that the terminal has a distance to the CPL larger than $\sqrt{6}R$, where $R$ is the radius of the single centralized LIS. Further splitting the 4 small LISs into 16 smaller LISs improves the CRLB, but may also increase the overheads of cooperating among different small LISs.

The rest of the paper is organized as follows. In Sec. II, we describe the signal propagation model for the LIS-terminal link and some common features of Fisher-information computations considered in the paper. In Sec. III, we derive the Fisher-information and CRLB for a terminal on the CPL. In Sec. IV, we discuss a terminal not on the CPL and put forth closed-form approximations of the Fisher-information and CRLB. Further, we also derive the CRLB for angles-of-arrival (AoA) and distance estimations. In Sec. V, we extend the signal model in Sec. II into having a common unknown phase rotation $\varphi$, caused by, e.g., the front-end circuits of the LIS and the terminal. In Sec. VI, we discuss the impacts on the CRLB of different deployments of the LIS. Numerical results are provided in Sec. VII and Sec. VIII summarizes the paper.

\begin{figure}[t]
\begin{center}
\vspace*{-2mm}
\hspace*{-3mm}
\scalebox{0.65}{\includegraphics{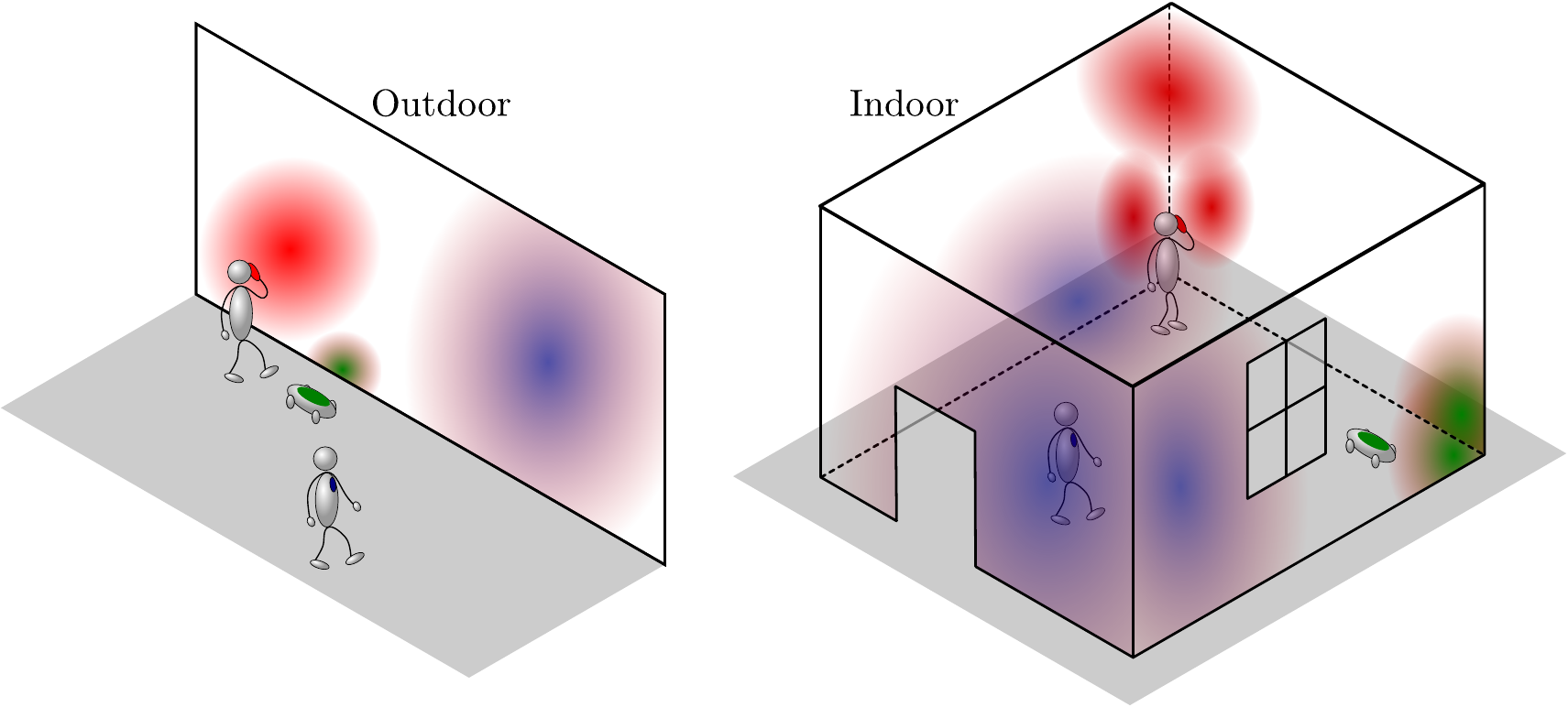}}
\vspace*{-6mm}
\caption{\label{fig1}Three users communicating with an LIS in an outdoor and an indoor scenarios.}
\vspace*{-12mm}
\end{center}
\end{figure}

\subsubsection*{Notation}
Throughout this paper, boldface letters indicate vectors and boldface uppercase letters designate matrices. Superscripts $(\cdot)^{-1}$, $(\cdot)^{1/2}$, $(\cdot)^{\ast}$, $(\cdot)\rmt$ and $(\cdot)\rmh$ stand for the inverse, matrix square root, complex conjugate, transpose, and
Hermitian transpose, respectively. In addition, $\mathcal{R}\{\cdot\}$ takes the real part. Moreover, throughout this paper, the asymptotic notation $f(\tau)\!=\!\mathcal{O}\left(g(\tau)\right)$ means that the limit of $f(\tau)/g(\tau)$ is equal to some constant as the parameter $\tau\!\to\!0$, while $B(\tau)\!=\!o\left(g(\tau)\right)$ means that the limit of $f(\tau)/g(\tau)$ goes to zero as $\tau\!\to\!0$.

\begin{figure}[b]
\begin{center}
\vspace*{-8mm}
\hspace*{0mm}
\scalebox{0.82}{\includegraphics{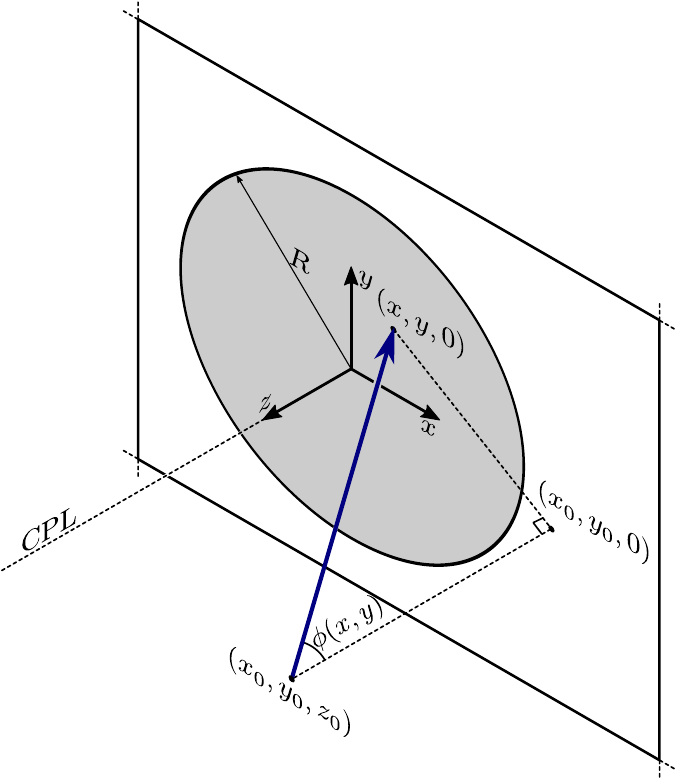}}
\vspace*{-5mm}
\caption{\label{fig2}The radiating model of transmitting signal to the LIS. We integrate the received signal of each point-element over the whole area spanned by the LIS. Therefore, for each point-element on the LIS, the Fraunhofer distance \cite{F10} is infinitely small and the received signal model (\ref{md}) holds for both near-filed and far-field scenarios.}
\vspace*{-14mm}
\end{center}
\end{figure}

\section{Signal Model with LIS}
Expressed in Cartesian coordinates, we assume that the center of the LIS is located at coordinates $x\!=\!y\!=\!z\!=\!0$ and a terminal is located at positive $z$-coordinate. The propagation model of a transmitting terminal at location $(x_0,y_0,z_0)$ to the LIS is depicted in Fig. \ref{fig2}. For analytical tractability, we assume an ideal situation where no scatterers or reflections are present, yielding a perfect LoS propagation scenario, and each terminal is assumed to radiate isotropically. Denoting the wavelength as $\lambda$, and assuming a narrow-band system and ideal free-space propagation from the terminal to all points at the LIS, the received signal at the surface at location $(x,\,y,\,0)$ radiated by a terminal at location $(x_0,\,y_0,\,z_0)$ is
\bea \label{md} \hat{s}_{x_0,\,y_0,\,z_0}(x,y)=s_{x_0,\,y_0,\,z_0}(x,y)+n(x,y), \eea
where $n(x,y)$ is modeled as zero-mean white Gaussian noise with spectral density $N_0$, and the noiseless signal $s_{x_0,\,y_0,\,z_0}(x,y)$ is stated in Property 1.
\begin{property} 
The noiseless signal $s_{x_0,\,y_0,\,z_0}(x,y)$ can be modeled as
\bea \label{md1} s_{x_0,\,y_0,\,z_0}(x,y)=\frac{\sqrt{z_0}}{2\sqrt{\pi}\eta^{\frac{3}{4}}}\exp\!\left(\!-\frac{2\pi j\sqrt{\eta}}{\lambda}\right)\!,   \eea
where the metric
\bea \eta\!=\!z_0^2\!+\!(y\!-\!y_0)^2\!+\!(x\!-\!x_0)^2.\eea 
\end{property} 
\begin{proof}
The noiseless signal received by the LIS at location $(x,y,0)$ and at time epoch $t$ as shown in Fig. \ref{fig2} reads,
\bea \label{nbs1} s_{x_0,\,y_0,\,z_0}(x,y)=\sqrt{P_L\cos\phi(x,y)}s(t) \exp\!\left(\!-2\pi j f_c \Delta_t(x,y)\right), \eea
where $P_L$ denotes the path-loss, $\phi(x,y)$ is AoA of the transmitted baseband signal $s(t)$ at $(x,y,0)$, and $f_c$ is the carrier-frequency. The transmit-time from the terminal to $(x,y,0)$ equals $ \Delta_t(x,y)=\frac{\sqrt{\eta}}{c}$, where $c$ is the speed-of-light. Since we are considering a narrow-band system, the signal $s(t)$ can be assumed to be the same at all locations $(x,y,0)$ of the LIS, hence we can let $s(t)\!=\!1$ and remove it from (\ref{nbs1}). Further, as the free-space path-loss $P_L\!=\!\frac{1}{4\pi\eta}$ and $\cos\phi(x,y)\!=\!\frac{z_0}{\sqrt{\eta}}$, inserting them back into (\ref{nbs1}) yields (\ref{md1}).
\end{proof}

In order to analyze the CRLBs for positioning, we denote the first-order derivatives of the noiseless signal in (\ref{md1}) with respect to variables $x_0$, $y_0$ and $z_0$ as $\Delta s_1$, $\Delta s_2$, and $\Delta s_3$, respectively, which are equal to
{\setlength\arraycolsep{2pt}  \bea \label{dev1} \Delta s_1&=&\frac{\sqrt{z_0}\left(x -x_0\right)}{2\sqrt{\pi}}\!\left(\!
\frac{3}{2}\eta^{-\frac{7}{4}} \!+\! \frac{2\pi j }{\lambda}{\eta}^{-\frac{5}{4}}\!\right)\!\exp\left(\!-\frac{2\pi j\sqrt{\eta}}{\lambda} \right)\!,   \\
\label{dev2} \Delta s_2&=&\frac{\sqrt{z_0}\left(y -y_0\right)}{2\sqrt{\pi}}\left(
\!\frac{3}{2}\eta^{-\frac{7}{4}} \!+\! \frac{2\pi j }{\lambda}{\eta}^{-\frac{5}{4}}\!\right)\!\exp\left(\!-\frac{2\pi j\sqrt{\eta}}{\lambda}\right)\!,    \\
\label{dev3} \Delta s_3&=&\frac{z_0^{\frac{3}{2}}}{2\sqrt{\pi}}\!\left(
\frac{1}{2z_0^2}\eta^{-\frac{3}{4}} \!- \!\frac{3}{2}\eta^{-\frac{7}{4}} \!- \!\frac{2\pi j }{\lambda}{\eta}^{-\frac{5}{4}}\!\right)\!\exp\left(\!-\frac{2\pi j\sqrt{\eta}}{\lambda} \right)\!.   \eea}
\hspace{-1.4mm}From \cite[Chapter 15]{K93}, as $\hat{s}_{x_0,\,y_0,\,z_0}(x,y)$ is Gaussian with mean $s_{x_0,\,y_0,\,z_0}(x,y)$ and variance $N_0$, the elements of the Fisher-information matrix are given by the following double integrals\footnote{The integrals in (\ref{Fisherij}) are due to the additive property of Fisher-information, which can be obtained by sampling the continuous signal $\hat{s}_{x_0,\,y_0,\,z_0}(x,y)$ with Nyquist frequency\cite{K93}, and the band-limited property of $s_{x_0,\,y_0,\,z_0}(x,y)$ can be seen from \cite{HRE171}.}
\bea \label{Fisherij} I_{ij}=\frac{2}{N_0}\iint_{x, y}\Re\!\left\{\Delta s_j \left(\Delta s_i\right)^{\ast}\right\}\mathrm{d}x\mathrm{d}y, \eea
where the integrals are taken over the area of the LIS, which we assume to have a disk-shape\footnote{The assumption of assuming a LIS with a disk-shape is solely for the convenience of derivations. Other shapes such as rectangular, triangular, or ring shapes can be analyzed in similar ways. Moreover, when the terminal is in the far-field, i.e., $R\!\ll\!z_0$, the shape of the LIS is irrelevant and can be regarded as a disk-shape with equal surface-area.} with radius $R$. As CRLB scales down linearly in SNR, we set $N_0\!=\!2$ throughout the paper to eliminate the scaling factor in (\ref{Fisherij}). Further, we define three functions $g_1(n)$, $g_2(n)$ and $g_3(n)$ which are necessary to compute the CRLB based on  (\ref{dev1})-(\ref{Fisherij}),
{\setlength\arraycolsep{2pt}\bea  \label{g1}  g_1(n)&=&\iint_{x^2+y^2\leq R^2}x^2\eta^{-\frac{n}{2}}\mathrm{d}x\mathrm{d}y,  \\
 \label{g2}  g_2(n)&=&\iint_{x^2+y^2\leq R^2}y^2\eta^{-\frac{n}{2}}\mathrm{d}x\mathrm{d}y, \\
\label{g3} g_3(n)&=&\iint_{x^2+y^2\leq R^2}\eta^{-\frac{n}{2}}\mathrm{d}x\mathrm{d}y.  \eea}
\hspace{-1.4mm}In general, closed-form expressions of $g_1(n)$, $g_2(n)$ and $g_3(n)$ are out of reach, except for the case that $x_0\!=\!y_0\!=\!0$, i.e., the terminal is on the CPL, and it holds that
{\setlength\arraycolsep{2pt} \bea \label{g1n} g_1(n)=g_2(n)&=&\frac{\pi}{n^2 - 6n + 8}\left(2z_0^{4 -n}\!- \!{\left(R^2 + z_0^2\right)}^{1 - \frac{n}{2}}\, \left(n R^2 - 2R^2 + 2 z_0^2\right) \right),  \\
\label{g2n} g_3(n)&=&\frac{z_0^{2 -n} - {\left(R^2 + z_0^2\right)}^{1 - \frac{n}{2}}}{n - 2}.  \eea}
\hspace{2.5mm}For a terminal that is not on the CPL, to analyze the properties of CRLB, we will use effective approximations for the Fisher-information and CRLB with the results obtained for the CPL case.

\section{CRLB of a Terminal on the CPL}
In this section, we analyze the CRLB for terminals along the CPL with coordinates $(0, 0, z_0)$. A nice property is that, the CRLB for all dimensions are in closed-form, i.e., the integrals (\ref{Fisherij}) can be efficiently solved using (\ref{g1n}) and (\ref{g2n}). We denote the Fisher-information and CRLB for a terminal with coordinates $(x_0, y_0,z_0)$ and a LIS with radius $R$ as $I_i([x_0, y_0,z_0], R)$ and $C_i([x_0, y_0,z_0], R)$, where the suffix $i=x, y, z$ represents the $x$, $y$, and $z$ dimension, respectively. When suffix $i$ has multiple variables, we mean that all these dimensions contained in $i$ are of the same value. For instance, $I_{x, y}([x_0, y_0,z_0], R)$ denotes the Fisher-information for both $x$ and $y$ dimensions whenever they are equal. 

We denote an useful parameter
\bea \label{tau} \tau=\left(R/z_0\right)^{2},\eea
which measures the relative surface-area normalized by the squared distance of the considered terminal position to the LIS. For a terminal in the far-field (in relation to the radius $R$), the value of $\tau$ is small, and for a terminal close to the LIS, $\tau$ becomes large.

\subsection{CRLB for Three Cartesian Dimensions}
\begin{theorem}
The Fisher-information matrix \vec{I} for a terminal with coordinates $(0,0,z_0)$ is diagonal and the Fisher-information for each Cartesian dimension is
{\setlength\arraycolsep{2pt} \bea  \label{Ixy}I_{x, y}([0,0,z_0], R)=\frac{1}{30z_0^2}f_1(\tau)+\frac{2\pi^2}{3\lambda^2}f_2(\tau),  \\
 \label{Iz}  I_z([0,\,0,\, z_0], R)=\frac{1}{40z_0^2}f_3(\tau)+\frac{2\pi^2}{3\lambda^2}f_4(\tau), \eea}
\hspace{-1.4mm}where the functions $f_1(\tau)$, $f_2(\tau)$, $f_3(\tau)$, and $f_4(\tau)$ obtained with (\ref{g1n})-(\ref{g2n}) are defined as
{\setlength\arraycolsep{2pt} \bea \label{f1} f_1(\tau)&=&1-\frac{1+2.5\tau}{(1+\tau)^{\frac{5}{2}}},   \\
\label{f2} f_2(\tau)&=&1-\frac{1+1.5\tau}{(1+\tau)^{\frac{3}{2}}},   \\
\label{f3} f_3(\tau)&=&13-\frac{13+5\tau^2}{(1+\tau)^{\frac{5}{2}}},   \\
\label{f4} f_4(\tau)&=&1-\frac{1}{(1+\tau)^{\frac{3}{2}}}.\eea}
\hspace{-1.4mm}respectively. Then, the CRLB for each dimension can be computed according to
\bea C_i([0,\,0,\, z_0], R)=I_{i}^{-1}([0,\,0,\, z_0], R), \; i=x, y, z. \eea
\end{theorem}
\begin{proof}
See Appendix A.
\end{proof}


From Theorem 1, the following conclusions can be derived. Firstly, when the terminal is close to the LIS, the Fisher-information is infinitely large for all dimensions, and the CRLB $C_i([0,\,0,\, z_0], R)$ becomes 0, while under the case $R\ll z_0$, the CRLB are $\infty$. These observations are consistent with the nature of the problem at hand. 

Secondly, in order to get a direct view of the CRLB in relation to surface-area of the LIS, we assume $\lambda\!\ll\! z_0$ (which in general holds as $\lambda$ is the wavelength). Then, the terms of the Fisher-information comprising $f_1(\tau)$ and $f_3(\tau)$ in (\ref{Ixy}) and (\ref{Iz}) can be omitted, and the CRLBs can be approximated as
{\setlength\arraycolsep{2pt} \bea \label{Cxy}C_{x,y}([0,\,0,\, z_0], R)
&\approx&\frac{3\lambda^2}{2\pi^2f_2(\tau)}, \\
 \label{Cz} C_z([0,\,0,\, z_0], R)
&\approx&\frac{3\lambda^2}{2\pi^2f_4(\tau)}. \eea}
\hspace{-1.4mm}respectively. As it can been seen that, the CRLB for all dimensions are uniquely decided by $\lambda$ and $\tau$. Hence, when $z_0$ is increased by a factor, the radius $R$ of the LIS also has to increase by the same factor in order to have the same CRLBs. Another interesting fact is that, the CRLBs for $x$ and $y$ dimensions are higher than that for $z$-dimension due to 
\bea f_2(\tau)\!<\!f_4(\tau),\eea
which can be seen directly from (\ref{f2}) and (\ref{f4}) using the fact $\tau\!>\!0$. 

Lastly, when the surface radius $R$ is much larger than the distance $z_0$ from the terminal to the LIS, it holds that
\bea  \lim_{\tau\to \infty} f_2(\tau)= \lim_{\tau\to \infty} f_4(\tau)=1.  \eea
Therefore, the asymptotic CRLBs in (\ref{Cxy}) and (\ref{Cz}) are identical and equal to
\bea \label{limtCRLB} \lim_{\tau\to \infty} C_{x, y, z}([0,\,0,\, z_0], R)=\frac{3\lambda^2}{2\pi^2},\eea
for all three dimensions and depend solely on the wavelength $\lambda$, which represents a fundamental lower limit to positioning precision.

In practical scenarios, a more interesting case is $R\!\ll\! z_0$, and then we have the following approximations by using Taylor expansions \cite{R87} at $\tau\!=\!0$, 
{\setlength\arraycolsep{2pt}\bea f_1(\tau)&=&\frac{15}{8}\tau^{2}+o\!\left(\tau^{2}\right),   \\ 
f_2(\tau)&=&\frac{3}{8}\tau^{2}+o\!\left(\tau^{2}\right),   \\ 
f_3(\tau)&=&\frac{65}{2}\tau+o\!\left(\tau\right),   \\ 
f_4(\tau)&=&\frac{3}{2}\tau+o\!\left(\tau\right).\eea}
\hspace{-1.4mm}respectively. From Theorem 1, the CRLBs for different dimensions are equal to
{\setlength\arraycolsep{2pt}\bea \label{area11} C_{x,y}([0,\,0,\, z_0], R)&=&16\tau^{-2}\left(\frac{1}{z_0^2}+\frac{4\pi^2}{\lambda^2 }\right)^{-1}+o\!\left(\tau^{-2}\right),   \\
\label{area22} C_z([0,\,0,\, z_0], R)&=&16\tau^{-1}\left(\frac{13}{z_0^2}+\frac{16\pi^2}{\lambda^2 }\right)^{-1}+o\!\left(\tau^{-1}\right). \eea}
\hspace{-1.4mm}As $\tau$ is proportional to $R^2$, for a terminal on the CPL the CRLBs for $x$ and $y$ dimensions decreases quadratically in the surface-area, while the CRLB for $z$-dimension decreases linearly in the same. Moreover, if we also assume that $\lambda\!\ll\! z_0$ (which usually holds as $\lambda$ is the wavelength), the CRLBs in (\ref{area11}) and (\ref{area22}) can be effectively approximated as
{\setlength\arraycolsep{2pt}\bea \label{area1} C_{x,y}([0,\,0,\, z_0], R)&\approx&\frac{4 \lambda^2}{\pi^2 \tau^2},   \\
\label{area2} C_z([0,\,0,\, z_0], R)&\approx&\frac{\lambda^2}{\pi^2\tau}, \eea}
\hspace{-1.8mm}respectively, which only depend on $\lambda$ and $\tau$. As will be shown in the next section, when the terminal moves away from the CPL, the CRLBs for all three dimensions degrade dramatically compared to (\ref{area1}) and (\ref{area2}), and decreases quadratically in the surface-area.

\section{CRLB of a Terminal not on the CPL}
In this section, we consider a terminal with arbitrary coordinates $(x_0, y_0, z_0)$. When $x_0, y_0\!\neq\!0$, closed-form expressions of the CRLB seem out of reach due to the complicated integrals in (\ref{Fisherij}). Therefore, we seek approximations, tight enough so that insights can still be drawn, of the CRLBs. Using the closed-form expressions of Fisher-information for a terminal on the CPL in Sec. III, the CRLBs for general cases can be well approximated as elaborated next.

\subsection{CRLB Approximations for a Terminal with Coordinates $(x_0, y_0, z_0)$}
We first introduce two mild conditions\footnote{These two conditions are only used to simplify the expressions (\ref{dev1})-(\ref{dev3}). That is, only the terms containing $1/\lambda$ in (\ref{dev1})-(\ref{dev3}) are preserved and the remaining terms are omitted since they are negligible compared to other terms comprising $1/\lambda$, which simplifies the calculations of Fisher-information as shown later in Appendix B.},
 {\setlength\arraycolsep{2pt} \bea \label{cond1}  \lambda&\ll& \frac{z_0^2}{\sqrt{z_0^2+x_0^2+y_0^2+R^2}}, \\
  \label{cond2} 2R&\ll&\frac{z_0^2}{\sqrt{x_0^2+y_0^2}}+\sqrt{x_0^2+y_0^2}.\eea}
\hspace{-1.4mm}As for the cases of interest $R$ is relatively small compared to $z_0$, and $\lambda$ is much smaller than $z_0$, these two conditions are usually satisfied. Letting
\bea z_1=\sqrt{x_0^2 + y_0^2 + z_0^2},\eea
the approximations for Fisher-information and CRLB matrices are stated in Property 2.

\begin{property}
Under the conditions (\ref{cond1})-(\ref{cond2}), the Fisher-information matrix for a terminal with coordinates $(x_0, y_0, z_0)$ can be approximated as
{\setlength\arraycolsep{5pt}  \bea  \label{Imat} \vec{I}=\left[\!\begin{array}{ccc} \alpha+\frac{\beta\, x_0^2}{z_0^2}  & \frac{\beta\, x_0\, y_0}{z_0^2} & \frac{\beta\, x_0}{z_0}\\ \frac{\beta\, x_0\, y_0}{z_0^2} & \alpha+\frac{\beta\, y_0^2}{z_0^2}  & \frac{\beta\, y_0}{z_0}\\ \frac{\beta\, x_0}{z_0} & \frac{\beta\, y_0}{z_0} &\beta \end{array}\right]\!\left(1+o\left(\frac{\lambda}{z_1}\right)\right)\!, \eea}
\hspace{-1.4mm}where $\alpha$ and $\beta$ are equal to
{\setlength\arraycolsep{2pt} \bea \label{alpha} \alpha&=& \frac{z_0}{z_1}I_{x,y}([0, 0, z_1], R),\\
 \label{beta} \beta&=& \left(\frac{z_0}{z_1}\right)^{\!\!3}I_z([0, 0, z_1], R), \eea}
\hspace{-1.4mm}and $I_{x,y}([0, 0, z_1], R)$, $I_z([0, 0, z_1], R)$ are the Fisher-information for $x, y$ and $z$ dimensions for a terminal with coordinates $(0, 0, z_1)$ that are stated in Theorem 1. Then the CRLB matrix reads
{\setlength\arraycolsep{5pt}  \bea  \label{Cmat} \vec{C}=\vec{I}^{-1}\approx\left[\!\begin{array}{ccc} \frac{1}{\alpha} & 0 & -\frac{x_0}{\alpha\, z_0}\\ 0 & \frac{1}{\alpha} & -\frac{y_{0}}{\alpha\, z_0}\\ -\frac{x_{0}}{a\, z_0} & -\frac{y_{0}}{\alpha\, z_0} &\frac{1}{ \beta} + \frac{ x_0^2+y_0^2}{\alpha z_0^2} \end{array}
\right]\!. \eea}
\end{property}
\begin{proof} 
See Appendix B.
\end{proof}

From Property 2, the Fisher-information and CRLB are approximated in closed-form. As a special case, when $x_0\!=\!y_0\!=\!0$, i.e., the terminal is on the CPL, the approximation (\ref{Cmat}) is exact as from Theorem 1. Further, we have the below corollary.
\begin{corollary}
Under the conditions (\ref{cond1})-(\ref{cond2}), the CRLBs for $x$ and $y$ dimensions are approximately equal, and depend on $(x_0, y_0, z_0)$ through $z_0$ and $\sqrt{x_0^2+y_0^2}$. That is, terminals on the circle $x_0^2\!+\!y_0^2\!=\!r^2$ have the same CRLBs for all dimensions for a given distance $z_0$.
\end{corollary}

Applying (\ref{area1})-(\ref{area2}) to approximate $I_{x,y}([0, 0, z_1], R)$ and $I_z([0, 0, z_1], R)$ in  (\ref{alpha})-(\ref{beta}), we have the approximated CRLBs for the non-CPL case stated in Property 3.

\begin{property}
Under the case $R\!\ll \!z_0$ and with conditions in (\ref{cond1})-(\ref{cond2}), the CRLBs for a terminal with coordinates $(x_0, y_0, z_0)$ can be approximated as
{\setlength\arraycolsep{2pt}\bea  \label{app1} C_{x,y}&\approx&\frac{4\lambda^2z_1^5}{\pi^2z_0R^4},   \\
  \label{app2} C_{z}&\approx&\frac{\lambda^2z_0^2}{\pi^2R^2}+\frac{4\lambda^2(x_0^2+y_0^2)z_1^5}{\pi^2z_0^3R^4}. \eea}
\end{property}

Compared to the CPL case, with a small $R$ the CRLB for $z$-dimension is dramatically degraded when the terminal is away from the CPL, that is, $x_0^2\!+\!y_0^2\!>\!0$. Further, when $\sqrt{x_0^2\!+\!y_0^2}\!>\!z_0$, the CRLB for $z$-dimension becomes even larger than the CRLBs for $x$ and $y$ dimensions. Furthermore, the CRLBs decrease quadratically\footnote{This is a consequence of the increasing CRLB for a terminal not on the CPL. As the limits of the CRLB when $R$ is $\infty$ are the same for a terminal at any position with the same $z_0$, the CRLB for a terminal located not the CPL must decrease faster than when it is located on the CPL.} in the surface-area of the LIS for all three dimensions in this case, which is an important motivation to go beyond the massive MIMO deployment to the LIS, which provides significant gains (quadratical in the surface-area) of the CRLB for positioning a terminal.

\subsection{CRLB for AoA and Radius Estimations}
Instead of estimating the coordinates $(x_0, y_0, z_0)$, in some cases is of interest to estimate the AoA and the distance $z_1$, in which case, the spherical coordinates reads
{\setlength\arraycolsep{2pt}\bea \label{spcart} 
x_0&=&z_1\sin\phi\cos\psi,\notag \\
y_0&=&z_1\sin\phi\sin\psi, \notag \\
z_0&=&z_1\cos\phi. \eea}
\hspace{-1.4mm}Define a vector function $\vec{g}(x_0,y_0,z_0)$ such that
\bea [z_1,\phi, \psi]=\vec{g}(x,y,z)=\left[g_1(x_0,y_0,z_0), g_2(x_0,y_0,z_0), g_3(x_0,y_0,z_0)\right]. \eea
where
{\setlength\arraycolsep{2pt} \bea g_1(x_0, y_0, z_0)&=&z_1,  \\
 g_2(x_0, y_0, z_0)&=&\arcsin\left(\frac{x_0}{z_1\cos\psi}\right), \\
 g_3(x_0, y_0, z_0)&=&\arctan\left(\frac{y_0}{x_0}\right).\eea}
\hspace{-1.4mm}Then, the Jacobian matrix $\nabla{\vec{g}}$ with respect to $(x_0, y_0, z_0)$ equals
\bea \nabla{\vec{g}}=\left[\nabla g_1\rmt,\nabla g_2\rmt, \nabla g_3\rmt\right]\rmt \eea 
where
{\setlength\arraycolsep{2pt} \bea \nabla g_1=\frac{\partial g_1}{\partial(x_0, y_0, z_0)}&=&\frac{1}{z_1}\left[x_0, y_0, z_0\right], \notag \\  
\nabla g_2=\frac{\partial g_2}{\partial(x_0, y_0, z_0)}&=&\frac{|\cos\psi|}{z_1^2\sqrt{z_1^2\cos^2\psi-x_0^2}\cos\psi}\left[z_1^2-x_0^2, -x_0y_0,  -x_0z_0\right], \notag \\  
\nabla g_3=\frac{\partial g_3}{\partial(x_0, y_0, z_0)}&=&\frac{1}{x_0^2+y_0^2}\left[-y_0, x_0, 0\right].
 \eea}
\hspace{-1.4mm}From \cite{K93}, the CRLB matrix for estimating $(z_1,\phi, \psi)$ equals
\bea \label{invI} \vec{C}&=&\nabla{\vec{g}}\vec{I}^{-1} (\nabla{\vec{g}})\rmh, \eea
and the CRLB for each parameter can be shown to be, after some manipulations,
{\setlength\arraycolsep{2pt}\bea   C_{z_1}&=&\frac{y_0^2-x_0^2}{z_1^2}C_x+\frac{z_0^2}{z_1^2}C_z,   \\
C_{\phi}&=&\frac{(z_1^4-x_0^4+x_0^2y_0^2)C_x+x_0^2z_0^2C_z}{z_1^4\left(z_1^2\cos^2\psi-x_0^2\right)}, \\
C_{\psi}&=&\frac{1}{x_0^2+y_0^2}C_x, \eea}
\hspace{-1.4mm}where $C_x$, $C_z$ are given in Property 2. Therefore, the CRLBs $C_{z_1}$, $C_{\phi}$, and $C_{\psi}$ in general also decrease quadratically in the surface-area. As a special case, if we consider a terminal on the CPL, that is, $\phi\!=\!\psi\!=\!0$, it holds that
{\setlength\arraycolsep{2pt}\bea  C_{z_1}&=&C_{z},  \\
C_{\phi}&=&\frac{1}{z_1^2}C_{x}.\eea}
\hspace{-1.4mm}However, $\phi\!=\!\psi\!=\!0$ is a singularity point for $C_{\psi}$.

\section{CRLB with Phase Uncertainty in Analog Circuits of the LIS}

In practical scenarios, the front-end circuitry of the LIS and of the terminal is not ideal and presents unknown distortions to the signal model. Using off-line calibration of the LIS \cite{JF17}, the entire LIS can be calibrated up to a common constant which is unknown to the LIS. The terminal has its own distortion, but what comes into play here is the product of the two distortions, which is then a single scalar number. In this paper we will model this distortion as a random phase uncertainty $\varphi$ since amplitude stability is easier to achieve in practice, see \cite{JF17}. Such a presence of the unknown phase $\varphi$ degrades the CRLB of positioning, and in this section we analyze the ensuing CRLB uncertainty thoroughly. To simplify the analysis, we take a special interest for a terminal on the CPL, while for the other positions we use numerical simulations.

With an unknown phase $\varphi$, the noiseless signal in (\ref{md1}) is modified to
\bea \label{mdpn} \tilde{s}_{x_0,\,y_0,\,z_0}(x,y)=\frac{\sqrt{z_0}}{2\sqrt{\pi}\eta^{3/4}}\exp\!\left(\!j\left(\!-\frac{2\pi \sqrt{\eta}}{\lambda}-\varphi\right)\!\right)\!.  \eea
Similarly, we denote the first-order derivatives with respect to variables $x_0$, $y_0$, $z_0$ and $\varphi$ as $\Delta \tilde{s}_1$, $\Delta \tilde{s}_2$, $\Delta \tilde{s}_3$, and $\Delta \tilde{s}_4$, respectively, which are
{\setlength\arraycolsep{2pt}  \bea \label{dev123pn} \Delta \tilde{s}_i&=& \Delta s_i\!\exp\left(-j\varphi\right),\;1\leq i\leq 3,  \\
\label{dev4pn} \Delta \tilde{s}_4&=&-j\tilde{s}_{x_0,\,y_0,\,z_0}(x,y),   \eea}
\hspace{-1.4mm}where $\Delta s_i$ are given in (\ref{dev1})-(\ref{dev3}). As the received signal $\hat{s}_{x_0,\,y_0,\,z_0}(x,y)$ is still Gaussian with mean $\tilde{s}_{x_0,\,y_0,\,z_0}(x,y)$ and variance $N_0$, the elements of Fisher-information matrix are still given by the double integrals in (\ref{Fisherij}). However, compared to the case without $\varphi$, in this case the Fisher-information matrix is 4-dimensional and the CRLBs for all three Cartesian dimensions are degraded. We then state the Fisher-information matrix for the non-CPL case in Theorem 2.

\begin{theorem}
With an unknown phase $\varphi$ considered in (\ref{mdpn}), the Fisher-information matrix equals
{\setlength\arraycolsep{5pt}  \bea  \label{Imatpn} \vec{I}=\left[\!\begin{array}{cc} \vec{I}_0&\vec{i}\rmt\\\vec{i}&I_{44}\end{array}\right]\!, \eea}
\hspace{-1.4mm}where $\vec{I}_0$ is the Fisher-information for $x, y$ and $z$ dimensions for the case with known phase $\varphi$, and the vector \vec{i} comprises the cross-terms of Fisher-information between the $x, y$, $z$ dimensions and the phase $\varphi$, which equals
\bea \vec{i}=\left[\!\begin{array}{ccc}\!I_{14}&\!I_{24}&\!I_{34}\!\!\end{array}\right]\!=\frac{z_0g_3(4)}{\lambda}\left[\!\begin{array}{ccc}\!x_0&\!y_0&\!z_0\!\!\end{array}\right]\!. \eea
Further, the Fisher-information for the unknown $\varphi$ equals
\bea  \label{I44} I_{44}=\frac{z_0}{4\pi}g_3(3),\eea
where $g_3(n)$ is the integral defined in (\ref{g3}).
\end{theorem}
\begin{proof}
See Appendix C.
\end{proof}
From Theorem 2, if we know the CRLB matrix $\vec{C}_0\!=\!\left( \vec{I}_0\right)^{-1}$ for $x$, $y$ and $z$ dimensions for the case with known $\varphi$, the CRLB matrix with $\varphi$ can be computed as
 {\setlength\arraycolsep{5pt}\bea  \label{Cmatpn} \vec{C}=\frac{1}{I_{44}-\vec{i}\vec{C}_0\vec{i}\rmt}\left[\!\begin{array}{cc} \vec{C}_0\left(I_{44}-\vec{i}\vec{C}_0\vec{i}\rmt\right)+\vec{C}_0\vec{i}\rmt\vec{i}\vec{C}_0&-\vec{C}_0\vec{i}\rmt \\ -\vec{C}_0\vec{i}&1\end{array}\right]\!. \eea}
\hspace{-1.4mm}As can be seen from (\ref{Cmatpn}), the CRLB for estimating $\varphi$ equals
\bea \label{Cwpn0} C_{\varphi}=\frac{1}{I_{44}-\vec{i}\vec{C}_0\vec{i}\rmt},\eea
and the CRLB matrix for the three Cartesian dimensions becomes
\bea \label{C0pn}  \tilde{\vec{C}_0}=\vec{C}_0+\vec{C}_0\vec{i}\rmt\vec{i}\vec{C}_0C_{\varphi}.\eea
Hence, from (\ref{C0pn}) the CRLBs are dramatically degraded due to the presence of $\varphi$ for the three Cartesian dimensions with the additional term $\vec{C}_0\vec{i}\rmt\vec{i}\vec{C}_0C_{\varphi}$. However, as $\varphi$ plays no role in the Fisher-information matrix in (\ref{Imatpn}), we have the corollary below.
\begin{corollary}
The Fisher-information and CRLB for all three Cartesian dimensions and the phase $\varphi$ are independent of the true value of $\varphi$.
\end{corollary}

Since in general we cannot get $g_3(n)$ in closed-from, we start with analyzing the Fisher-information for a terminal on the CPL, which from Theorem 2 equals
{\setlength\arraycolsep{5pt} \bea  \label{Imatpn1} \vec{I}=\left[\!\begin{array}{cccc} I_{11}  &0 &0&0\\0&I_{22} &0 & 0\\ 0& 0 &I_{33}&I_{34}\\ 0&0&I_{34}&I_{44}\end{array}\right]\!\!. \eea}
\hspace{-1.4mm}Hence, the CRLBs for $x$ and $y$ dimensions remain the same with the unknown $\varphi$, and the CRLBs for $z$-dimension and phase $\varphi$ are equal to
{\setlength\arraycolsep{2pt}\bea \label{Czpn} C_z&=&\frac{I_{44}}{I_{33}I_{44}-I_{34}^2},  \\ 
\label{Cwpn} C_{\varphi}&=&\frac{I_{33}}{I_{33}I_{44}-I_{34}^2}.\eea}
\hspace{-1.4mm}On the CPL, we can reach expressions for $I_{34}$ and $I_{44}$ in closed-from, and with $I_{ii}$ ($1\!\leq\!i\!\leq\!3$) computed in Theorem 1, the CRLB for all dimensions are stated in the below property.
\begin{property}
With an unknown phase $\varphi$, for a terminal on the CPL the CRLBs for $x$ and $y$ dimensions remain the same as with known $\varphi$, while the CRLBs for $z$-dimension and phase $\varphi$ are equal to
 {\setlength\arraycolsep{2pt}\bea\label{Czpn1}   C_z&=&
\left(\frac{1}{10z_0^2}f_5(\tau)+\frac{\pi^2}{6\lambda^2}f_6(\tau)\right)^{\!\!-1}, \\
\label{Cwpn1} C_\varphi&=&\left(\frac{1}{2}f_7(\tau)+\left(\frac{\lambda^2}{10\pi^2z_0^2f_8(\tau)}+\frac{8}{3f_9(\tau)}\right)^{-1}\right)^{\!\!-1},\eea}
 \hspace{-1.4mm}where the functions $f_5(\tau)$, $f_6(\tau)$, $f_7(\tau)$, $f_8(\tau)$ and $f_9(\tau)$ obtained with (\ref{g1n})-(\ref{g2n}) are defined as
{\setlength\arraycolsep{2pt}\bea\label{f5}  f_5(\tau)&=&1-\frac{1+1.25\tau^2}{(1+\tau)^{\frac{5}{2}}},\\
\label{f6}  f_6(\tau)&=&1-\frac{4-3\sqrt{1+\tau}+3\tau}{(1+\tau)^{\frac{3}{2}}}, \\
\label{f7}  f_7(\tau)&=&1-\frac{1}{\sqrt{1+\tau}},\\
\label{f8}  f_8(\tau)&=&\frac{\tau^2\sqrt{1+\tau}}{4+5\tau^2-4(1+\tau)^{\frac{5}{2}}}, \\
\label{f9}  f_9(\tau)&=&\frac{\tau^2}{\sqrt{1+\tau}-(1+\tau)^2}.\eea}
\end{property}
\begin{proof}
See Appendix D.
\end{proof}

Using Property 4, when $\tau\!\to\!\infty$ it holds that
\bea  \lim\limits_{\tau\to\infty}f_5(\tau)=\lim\limits_{\tau\to\infty}f_6(\tau)=1, \eea
and the CRLB limit for $z$-dimension is
\bea \label{Czpnlimit1} \lim\limits_{\tau\to\infty} C_z=\frac{6\lambda^2}{\pi^2}, \eea
which is 4 times of the CRLB for $z$-dimension with known $\varphi$, hence, the unknown phase causes 6 dB degradation of the positioning precisions for $z$-dimension for a terminal on the CPL. Further, as it also holds that
{\setlength\arraycolsep{2pt}\bea  \lim\limits_{\tau\to\infty}f_7(\tau)&=&1,  \\
 \lim\limits_{\tau\to\infty}f_8(\tau)&=&-\frac{1}{4}, \\
  \lim\limits_{\tau\to\infty}f_9(\tau)&=&-1, \eea}
\hspace{-1.4mm}the CRLB limit for phase $\varphi$ equals
\bea  \label{Czpnlimit2} \lim\limits_{\tau\to\infty} C_\varphi=\left(\frac{1}{2}-\left(\frac{8}{3}+\frac{\lambda^2}{10\pi^2z_0^2}\right)^{-1}\right)^{\!\!-1}, \eea
which becomes a constant when $\lambda\!\ll\!z_0$,
\bea  \label{Czpnlimit2} \lim\limits_{\tau\to\infty} C_\varphi=8. \eea
Therefore, in order to estimate $\varphi$, the SNR should be extremely high regardless of the wavelength $\lambda$ and surface-area of the LIS.

To see the trends at small $\tau$, we also use Taylor expansions at $\tau\!=\!0$ which results in
{\setlength\arraycolsep{2pt}\bea\label{f51}  f_5(\tau)&=&\frac{5}{2}\tau+o\left(\tau\right),\\
\label{f61}  f_6(\tau)&=&\frac{1}{8}\tau^3+o\left(\tau^3\right), \\
\label{f71}  f_7(\tau)&=&\frac{1}{2}\tau+o(\tau),\\
\label{f81}  f_8(\tau)&=&-\frac{1}{10}\tau+o(\tau), \\
\label{f91}  f_9(\tau)&=&-\frac{2}{3}\tau+o(\tau).\eea}
\hspace{-1.4mm}From Property 4 and using (\ref{f51})-(\ref{f91}), when $\tau$ is sufficiently small we have the approximations
  {\setlength\arraycolsep{2pt}\bea\label{Czpn2}   C_z&\approx&\frac{48\lambda^2}{\pi^2\tau^3}\left(1+\frac{12\lambda^2}{\pi^2z_0^2\tau^2}\right)^{-1}, \\
\label{Cwpn2} C_\varphi&\approx&\frac{4}{\tau\lambda^2}\left(\lambda^2+4\pi^2z_0^2\right).\eea}
\hspace{-1.4mm}An interesting fact is that, unlike the case with known $\varphi$ where the CRLB for $z$-dimension decreases linearly in the surface-area, in the presence of an unknown $\varphi$ the slope of the CRLB for $z$-dimension in relation to the surface-area (both are in logarithmic domain) varies between 1 and 3. This can be seen from (\ref{Czpn2}) as we have the two cases:
 \begin{itemize}
 \item When $\frac{2\sqrt{3}\lambda}{\pi z_0}\!\ll\!\tau\!\ll\!1$, it holds that
\bea\label{Czpn3}   C_z\approx\frac{48\lambda^2}{\pi^2\tau^3},\eea
which decreases in the third-order of the surface-area of the LIS.
 \item When $0\!<\!\tau\!\ll\!\frac{2\sqrt{3}\lambda}{\pi z_0}$, it holds that
\bea\label{Czpn4}   C_z\approx\frac{4z_0^2}{\tau},\eea
which decreases linearly in the surface-area of the LIS.
 \end{itemize}
 
\begin{remark}
Note that, the CRLB for $z$-dimension in (\ref{Czpn4}) is independent of $\lambda$, which is different from the CRLB with known phase as in (\ref{area2}). Therefore, with phase uncertainty, decreasing the wavelength is not beneficial for improving the CRLB for estimating the distance $z_0$.
\end{remark}
 
Moreover, when $\tau$ is sufficiently small and $\lambda\!\ll\!z_0$ holds, the CRLB for phase $\varphi$ is significantly larger than that for $z$-dimension since
\bea\label{CzCwpn}   \frac{C_\varphi}{C_z}\approx\frac{4\pi^2}{\lambda^2}.\eea

In Fig. \ref{f1f2f3}, we depict the CRLB for all three Cartesian dimensions with and without $\varphi$, derived in Theorem 1 and Property 4, respectively, and we let $z_0\!=\!4$ m and $\lambda\!=\!0.1$ m. Assuming that the distance between two adjacent antenna elements in the surface-deployment is half of $\lambda$, the number of antenna-elements deployed in the surface is then equal to 
\bea N\!=\!\frac{4\pi R^2}{\lambda^2}\!=\!\frac{4\pi \tau z_0^2}{\lambda^2}\!=\!2\tau\!\times\!10^{4}.\eea 
A typical massive-MIMO array comprising $N\!=\!200$ antennas results in $\tau\!=\!0.01$. We see that massive-MIMO for positioning falls just short of reaching the cubic slope, whereas LIS that increases the surface-area 10-20 fold reaches the cubic slope and yields significant gains.

In Fig. \ref{z_omega_approx}, we depict the CRLBs for $z$-dimension and phase $\varphi$. As can be seen, the approximations in (\ref{Czpn2})-(\ref{Cwpn2}) are well aligned with the exact forms obtained in Property 4 when $\tau\!<\!0.02$. Moreover, in this case the CRLBs for estimating $\varphi$ is around $4\pi^2/\lambda^2\!=\!4000$ times of the CRLB for $z$-dimension which is shown in (\ref{CzCwpn}).

 \begin{figure}[t]
\begin{center}
\vspace*{-6mm}
\hspace*{-4mm}
\scalebox{0.42}{\includegraphics{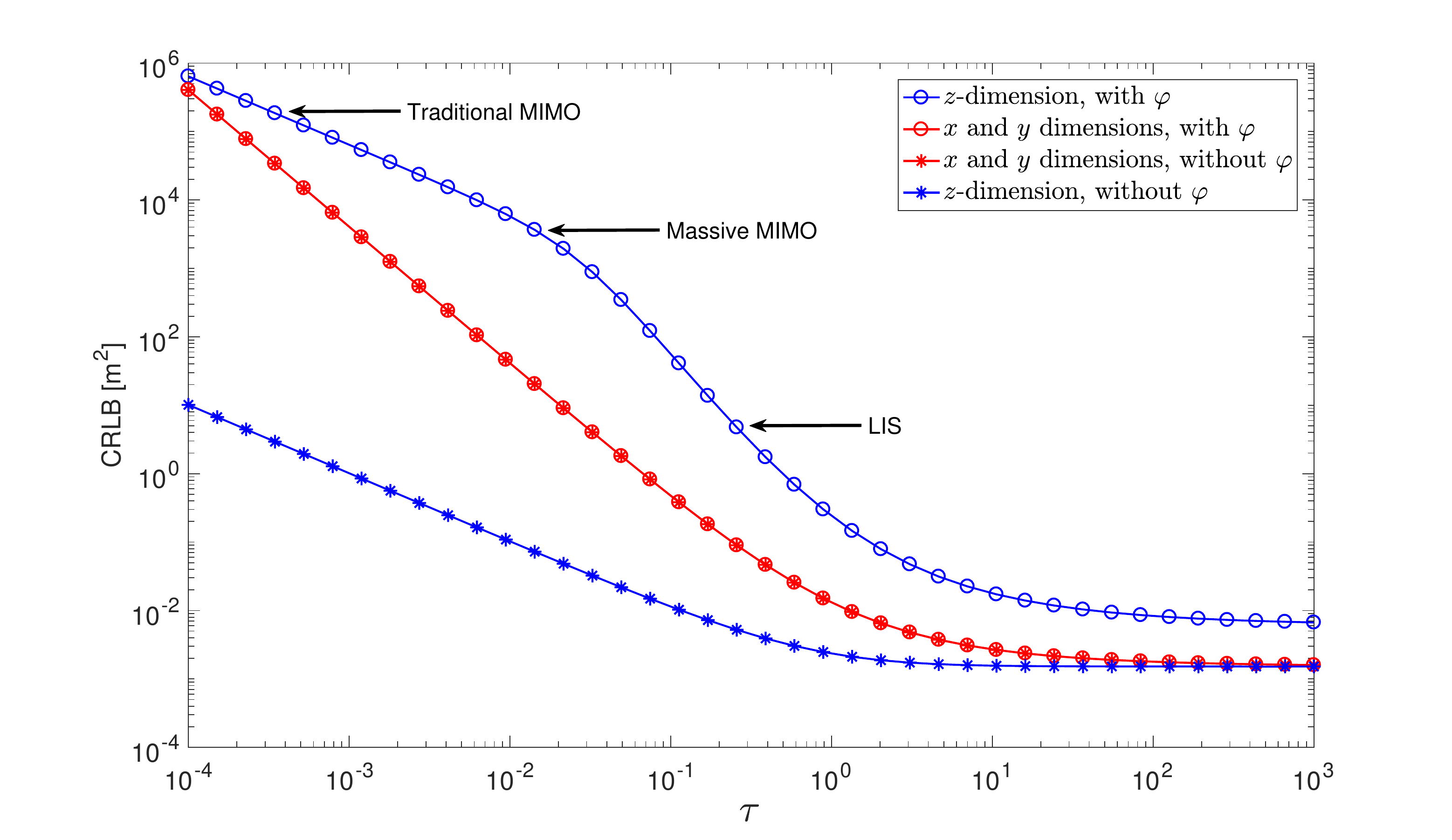}}
\vspace*{-8mm}
\caption{\label{f1f2f3}The exact CRLB for $x$, $y$ and $z$ dimensions for terminals along the CPL. As can be seen, with LIS the CRLB is the the fast-decreasing region compared to the massive-MIMO, which shows the potential gains of the LIS.}
\vspace*{-6mm}
\end{center}
\end{figure}

 \begin{figure}
\begin{center}
\vspace*{-4mm}
\hspace*{-4mm}
\scalebox{0.42}{\includegraphics{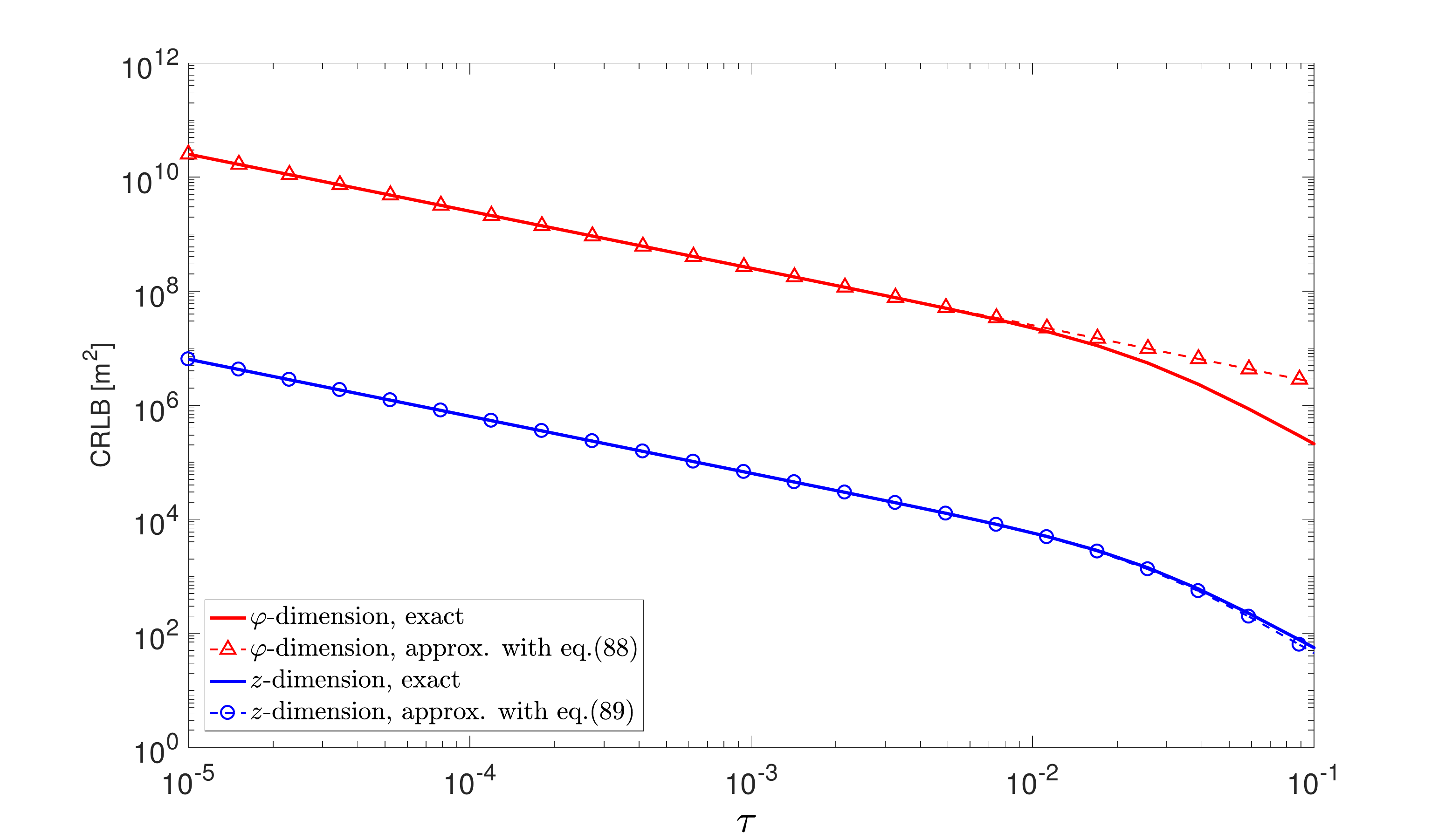}}
\vspace*{-8mm}
\caption{\label{z_omega_approx}The exact and approximated CRLB for $z$-dimension and phase $\varphi$ for terminals on the CPL, which are well aligned for small values of $\tau$.}
\vspace*{-12mm}
\end{center}
\end{figure}

Following the similar discussion in Sec. IV and utilizing the approximations in Property 2 and (\ref{Cwpn0})-(\ref{C0pn}), the CRLBs for a terminal not on the CPL can also be approximated. However, the derivations are relatively long and the conclusions are similar as those drawn for the case with a terminal on the CPL. Loosely speaking, the Fisher-information terms comprised in vector $\vec{i}$ increases linearly in the surface-area and $\vec{C}_0$ decreases quadratically in the surface-area for all dimensions as explained in Sec. IV. Further, as $\vec{C}_{\varphi}$ decreases linearly in the surface-area, the CRLBs for $x$, $y$ and $z$ dimensions then decrease in the third-order of the surface-area from (\ref{C0pn}), which is the same as for the CPL case. Furthermore, as $\tau$ grows large, the limits of the CRLBs for $x$ and $y$ dimensions remain the same for the case with known $\varphi$ since all positions can be approximated as on the CPL in the far-field, while for $z$-dimension the limit of CRLB is 6 dB higher than that with known $\varphi$ as shown in (\ref{Czpnlimit1}).

\section{Deployment of the LIS}
In this section we consider different deployments of the LIS on a large surface with size $W\!\times\! H$ where $W, H$ are the width and length, respectively. In particular, we consider the centralized-deployment (a) and distributed-deployments (b), (c) as depicted in Fig. \ref{fig4}. For simplicity, we assume $R,\,\lambda\!\ll\! z_0$ and consider the CRLBs for a terminal on the CPL with coordinates $(0, 0, z_0)$ without phase-uncertainty in the received signal, that is, positioning a terminal in the far-field.

For the centralized deployment (a), the CRLBs for all three dimensions are given in (\ref{area1}) and (\ref{area2}). With a distributed deployment (b), the LIS is split into four small LISs centered at $(\pm W/4, \pm H/4)$, each with radius $R/2$. Using Property 2, the symmetry of the LIS, and the approximations in (\ref{area1})-(\ref{area2}), the sum of the Fisher-information matrices corresponding to the four small LISs can be shown to be diagonal, and the Fisher-information for the $x, y$ and $z$ dimensions are equal to
{\setlength\arraycolsep{2pt} \bea I_{x,y}&\approx&\frac{\pi^2 z_0R^4}{16\lambda^2( z_0^2+D^2)^{5/2}}+\frac{\pi^2D^2 z_0R^2}{2\lambda^2( z_0^2+D^2)^{5/2}},  \\
I_z&\approx&\frac{\pi^2R^2 z_0^3}{\lambda^2( z_0^2+D^2)^{5/2}},  \eea}
\hspace{-1.4mm}respectively, where $D$ equals
\bea D\!=\!\frac{\sqrt{W^2+H^2}}{4}.\eea 
Assuming $D\!\ll \!z_0$, the Fisher-information can further be approximated as
 {\setlength\arraycolsep{2pt}    \bea  \label{Ixy4} I_{x,y}&\approx&\frac{\pi^2R^4}{4\lambda^2 z_0^4}\left(\frac{1}{4}+\frac{2D^2}{R^2}\right),  \\
 \label{Iz4} I_{z}&\approx&\frac{\pi^2R^2}{\lambda^2 z_0^2}. \eea}
\hspace{-1.4mm}Comparing (\ref{area1}) to (\ref{Ixy4}), it can be seen that the CRLB for $x$ and $y$ dimensions with the distributed deployment (b) is lower than that with the centralized deployment (a) only if 
\bea \frac{1}{4}+\frac{2D^2}{R^2}>1,\eea
 or equivalently, 
 \bea \label{thresh2}  \sqrt{W^2+H^2}>\sqrt{6}R. \eea
 
 \begin{figure*}[t]
\begin{center}
\vspace*{-6mm}
\hspace*{-2mm}
\scalebox{0.75}{\includegraphics{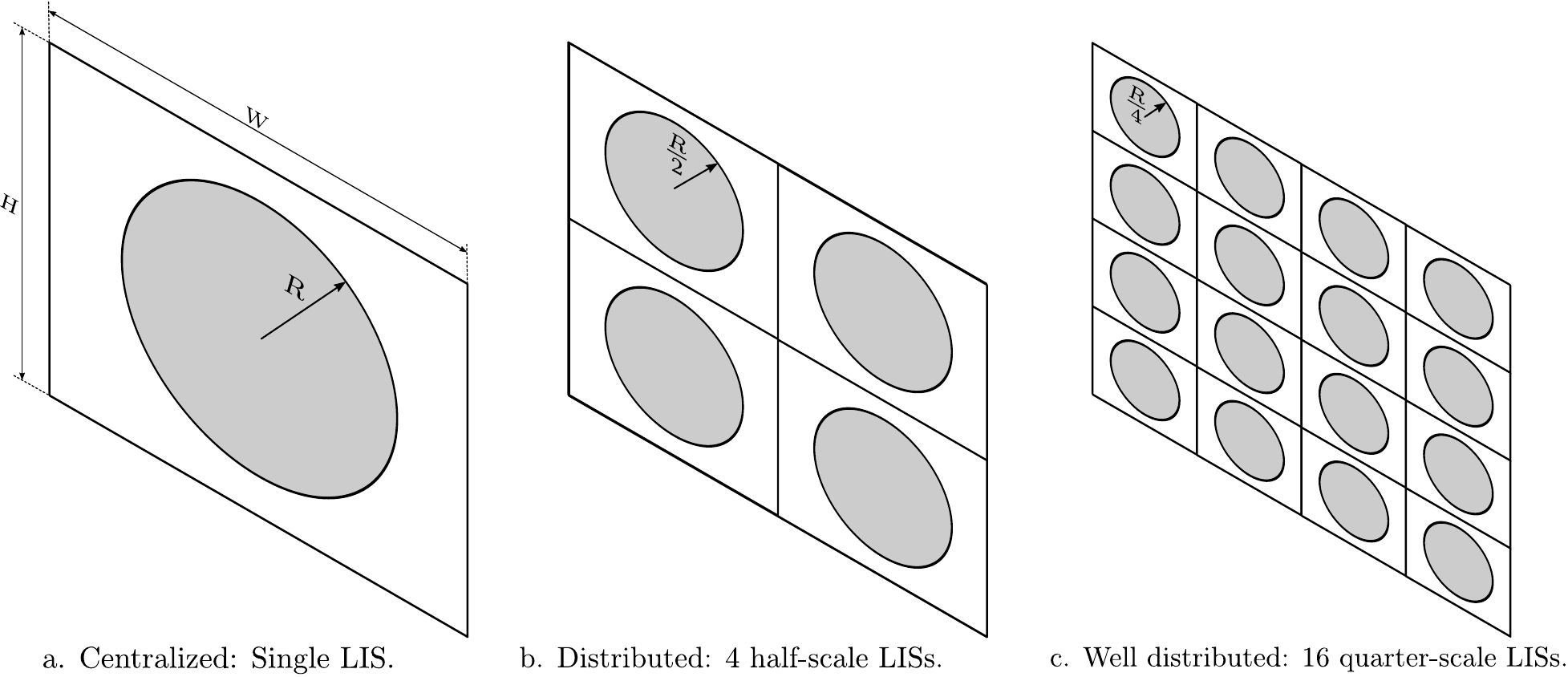}}
\vspace*{-6mm}
\caption{\label{fig4}Different deployments of the LIS in a surface with width $W$ and length $H$. Note that the total surface-area is the same for different deployments.}
\vspace*{-10mm}
\end{center}
\end{figure*}

That is to say, in the far-field with the distributed deployment (b), the CRLBs for $x$ and $y$ dimensions are improved if the four small LISs are deployed sufficiently far apart in relation to radius $R$. Otherwise, the centralized deployment (a) provides lower CRLBs for $x$ and $y$ dimensions than that for the distributed deployment. However, the CRLB for $z$-dimension remains the same for both deployments. Further, when $R\!\ll\!D$, the Fisher-information in (\ref{Ixy4}) becomes
 \bea  \label{Ixy5} I_{x,y}\approx\frac{\pi^2D^2R^2}{2\lambda^2 z_0^4}, \eea
which shows that, the CRLBs for $x$ and $y$ dimensions are not only improved, but also decreases linearly in the surface-area of the LIS with a distributed deployment rather than quadratically.

Following the same principle, one can split the LIS into more small pieces and obtain an ultra-densely distributed deployment such as in (c) of Fig. \ref{fig4}. In general, with a distributed deployment, the overall positioning performance is more robust than a centralized deployment, and the average positioning performance is improved which we show later with numerical simulations.

\begin{figure}[t]
\begin{center}
\vspace*{-6mm}
\hspace*{-4mm}
\scalebox{0.42}{\includegraphics{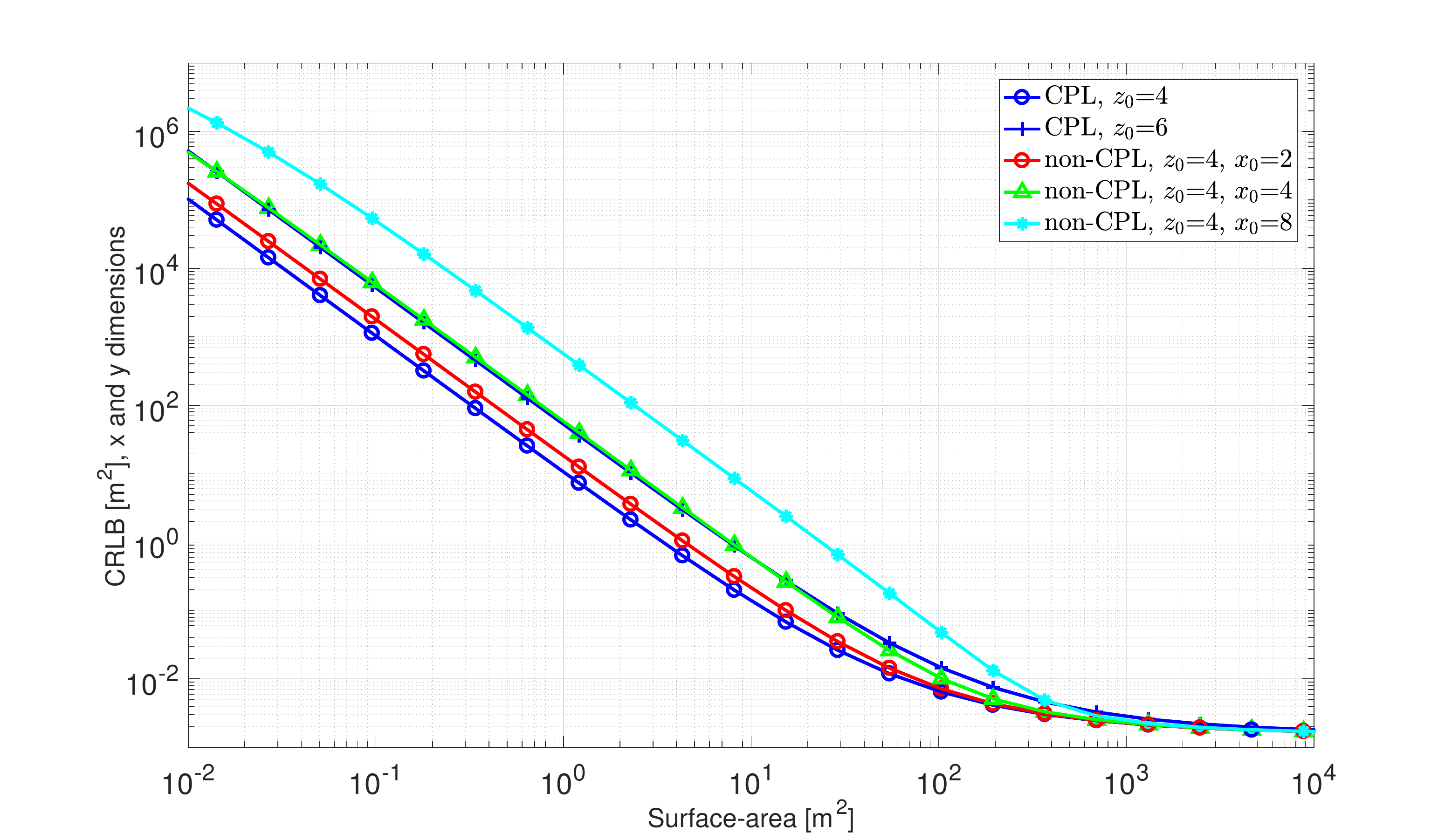}}
\vspace*{-8mm}
\caption{\label{fig6}CRLB for $x$ and $y$ dimensions, and the CRLBs for $y$-dimension are almost overlapped with those for $x$-dimension.}
\vspace*{-4mm}
\end{center}
\end{figure}

\begin{figure}
\begin{center}
\vspace*{-6mm}
\hspace*{-4mm}
\scalebox{0.42}{\includegraphics{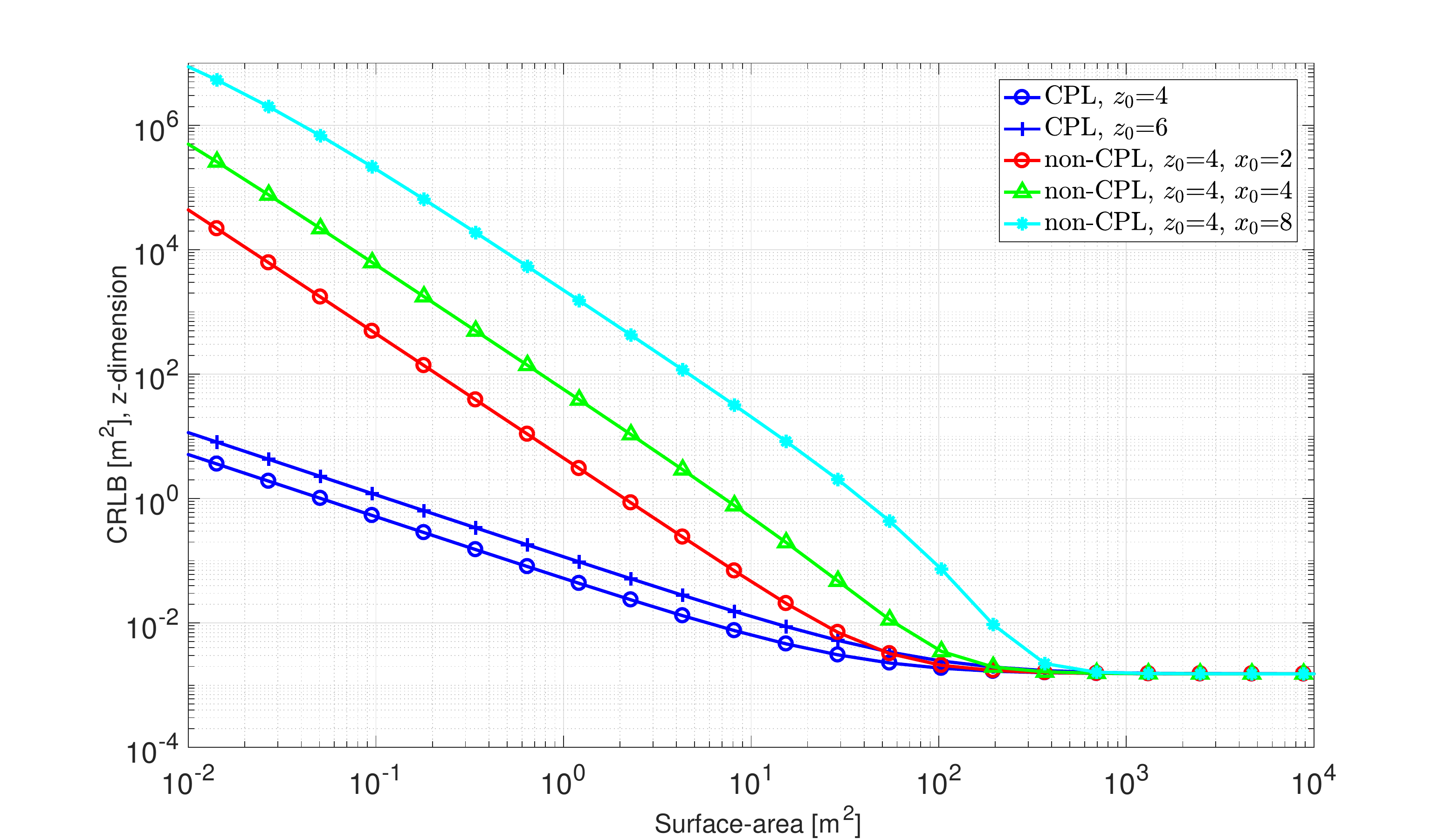}}
\vspace*{-8mm}
\caption{\label{fig7}CRLB for $z$-dimension with the same tests in Fig. \ref{fig6}.}
\vspace*{-12mm}
\end{center}
\end{figure}

\section{Numerical  Results}
In this section, numerical results are provided to illustrate the theories and conclusions that we have developed in previous sections. As explained earlier, in all tests we set the noise spectral-density to $N_0\!=\!2$, and without explicitly pointed out, the unit for the coordinates of the terminal, the wavelength $\lambda$ and the radius $R$ of the LIS are all in $m$, while the unit for CRLB is $m^2$.

\subsection{Exact-CRLB Evaluations}
We first evaluate the CRLB for terminals both on and away from the CPL as discussed in Sec. III and Sec. IV. As only the radius $\sqrt{x_0^2\!+\!y_0^2}$ matters as shown in Corollary 1, we illustrate with offsets only in $x$-dimension. In Fig. \ref{fig6} and Fig. \ref{fig7}, we test with $R\!=\!1$, $\lambda\!=\!0.1$, $y_0\!=\!0$, $x_0\!=\!2$, 4, 8, and $z_0\!=\!4$, 6, respectively, and some interesting results can be observed. 

Firstly, as shown in Fig. \ref{fig6}, when $\tau$ is small the CRLBs for $x$ and $y$ dimensions decrease quadratically in the surface-area of the LIS, while as shown in Fig. \ref{fig7}, the CRLB for $z$-dimension decreases only linearly in that. This is well aligned with the results in (\ref{area1}) and (\ref{area2}). Secondly, the CRLB for $z$-dimension increases dramatically when the terminal is away from the CPL. Furthermore, as long as $x_0 \!\neq \!0$, the CRLB for $z$-dimension also decreases quadratically in the surface-area. These phenomenons are well predicted by Property 2. Lastly, it can been seen that, as $R\!\to\!\infty$ the CRLB converges to a limit  $\frac{3\lambda^2}{2\pi^2}\!=\! 1.5\!\times\!10^{-3}$ for all dimensions as shown in (\ref{limtCRLB}).

\begin{figure}[t]
\begin{center}
\vspace*{-6mm}
\hspace*{-4mm}
\scalebox{0.42}{\includegraphics{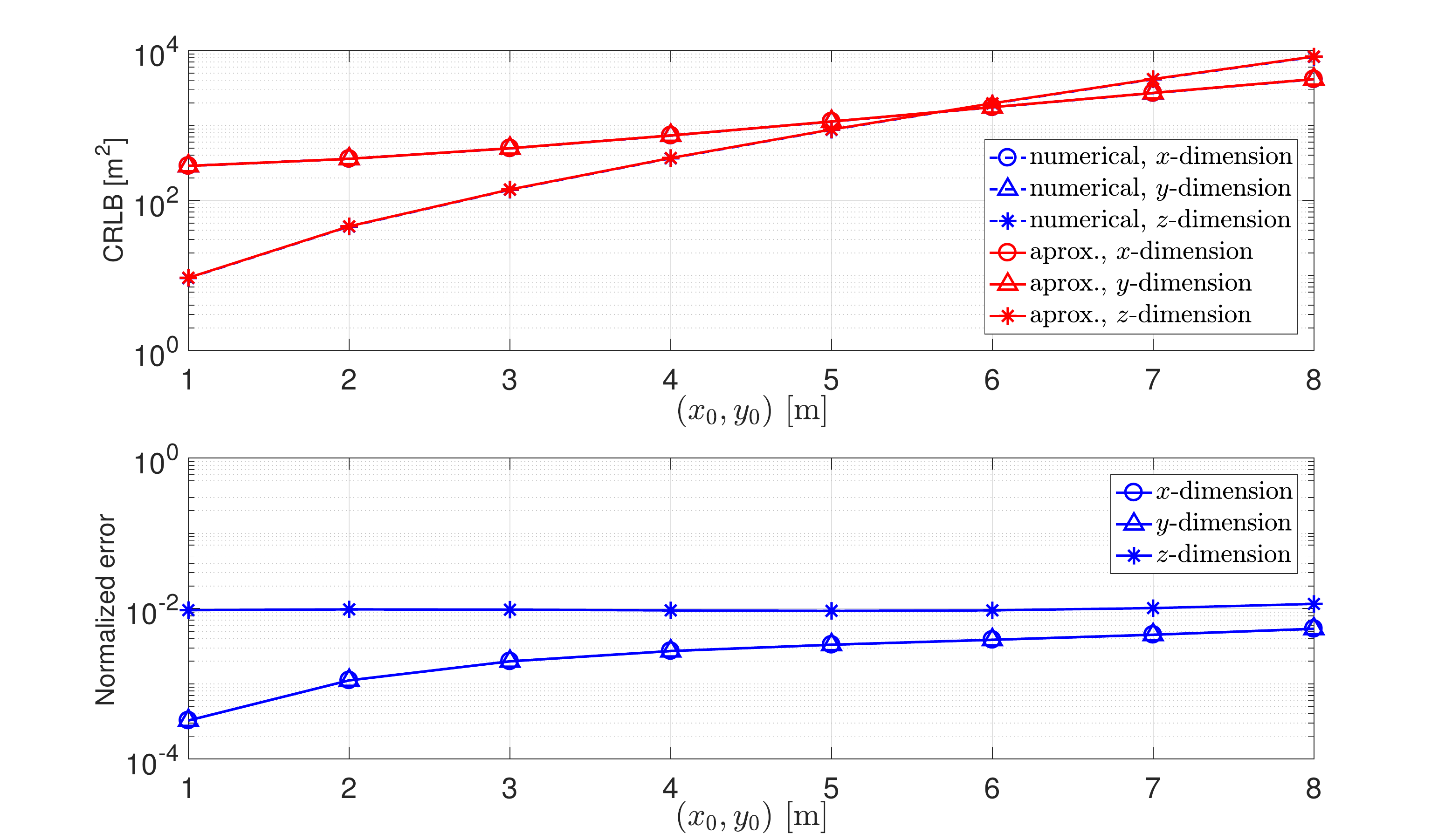}}
\vspace*{-8mm}
\caption{\label{fig8}CRLB computed with numerical integrations and their approximations using (\ref{app1})-(\ref{app2}) in Property 2, and the normalized approximation errors.}
\vspace*{-6mm}
\end{center}
\end{figure}

\begin{figure}
\begin{center}
\vspace*{-4mm}
\hspace*{-4mm}
\scalebox{0.42}{\includegraphics{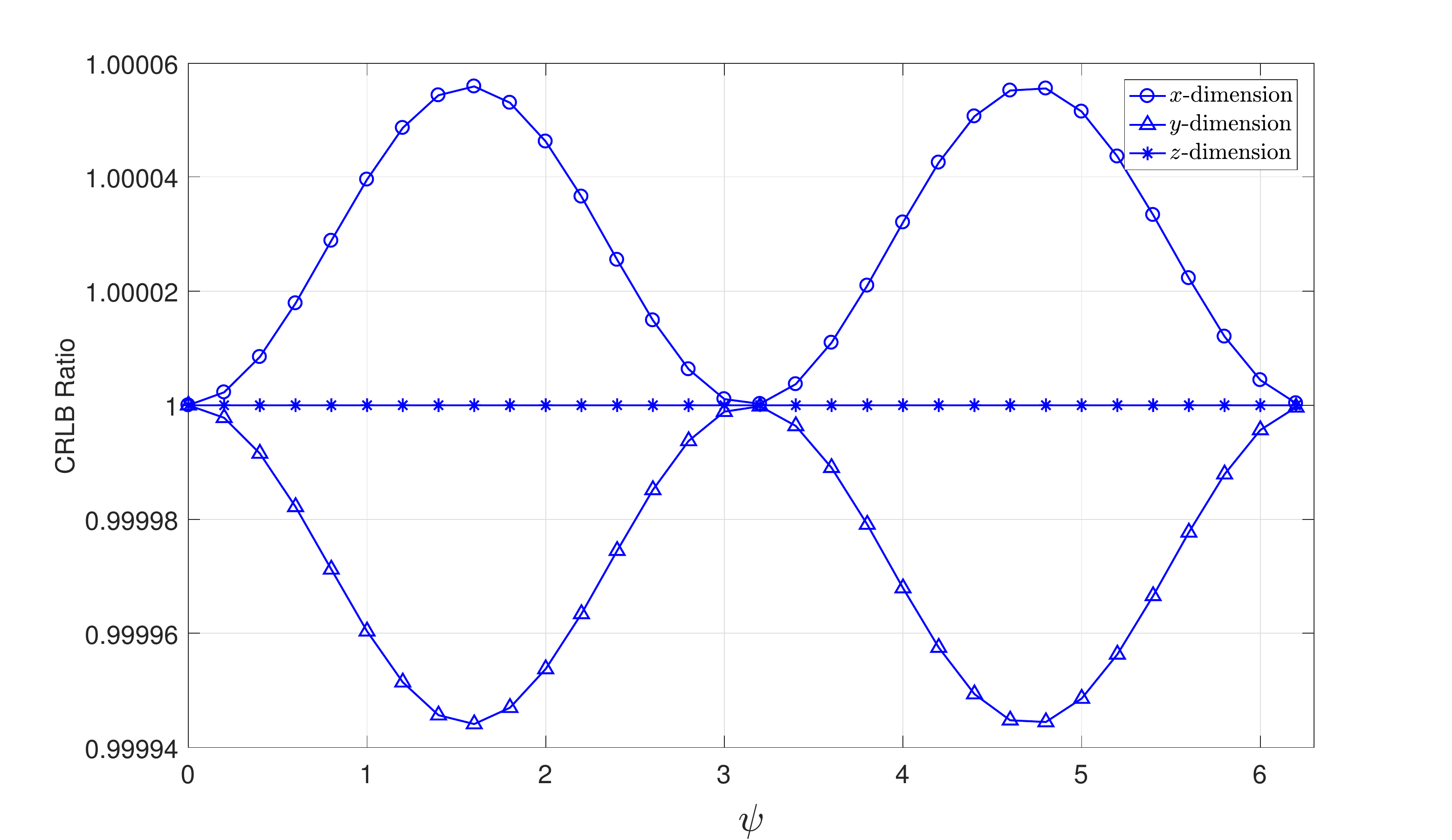}}
\vspace*{-8mm}
\caption{\label{fig9}The CRLB differences for terminals on a circle that is parallel to the LIS with center $(0, 0, 4)$ and radius $r\!=\!4$.}
\vspace*{-12mm}
\end{center}
\end{figure}

\subsection{CRLB-Approximation Accuracies}

Next, we evaluate the CRLB approximations for a terminal not on the CPL as discussed in Sec. IV. We compare the numerical integration results\footnote{For numerical computation of the CRLB, we use the Matlab built-in function \lq{}integral\rq{} to calculate the integrals in the CRLB matrix directly, which has an absolute error of $10^{-10}$ and relative error of $10^{-6}$.} of CRLB and their approximations using (\ref{app1})-(\ref{app2}) in Property 2. We test with $R\!=\!0.5$, $\lambda\!=\!0.1$, $z_0\!=\!8$, and $x_0\!=\!y_0$ varying from 1 to 8.

The CRLBs and the normalized approximation errors that are computed as the normalized CRLB differences between the numerical integrations and the approximations are both shown in Fig. \ref{fig8}. As can be seen, the approximations given by Property 2 perform well, with normalized errors less than 0.5\% for the $x$ and $y$ dimensions, and close to 1\% for $z$-dimension.

In Fig. \ref{fig9}, we repeat the tests in Fig. \ref{fig6} and Fig. \ref{fig7} with numerical integrations, but setting $x_0\!=\!r\cos\psi$ and $y_0\!=\!r\cos\psi$, with $r\!=\!4$ and $\psi$ changing over $[0, 2\pi]$. The CRLBs in all dimensions are normalized with those obtained at coordinates $(4, 0, 8)$. As can be seen, the CRLB for the $z$-dimension is identical for any angle $\psi$, while the CRLBs for $x$ and $y$ dimensions are almost identical. These observations corroborate Corollary 1.

\begin{figure}[t]
\begin{center}
\vspace*{-6mm}
\hspace*{-4mm}
\scalebox{0.42}{\includegraphics{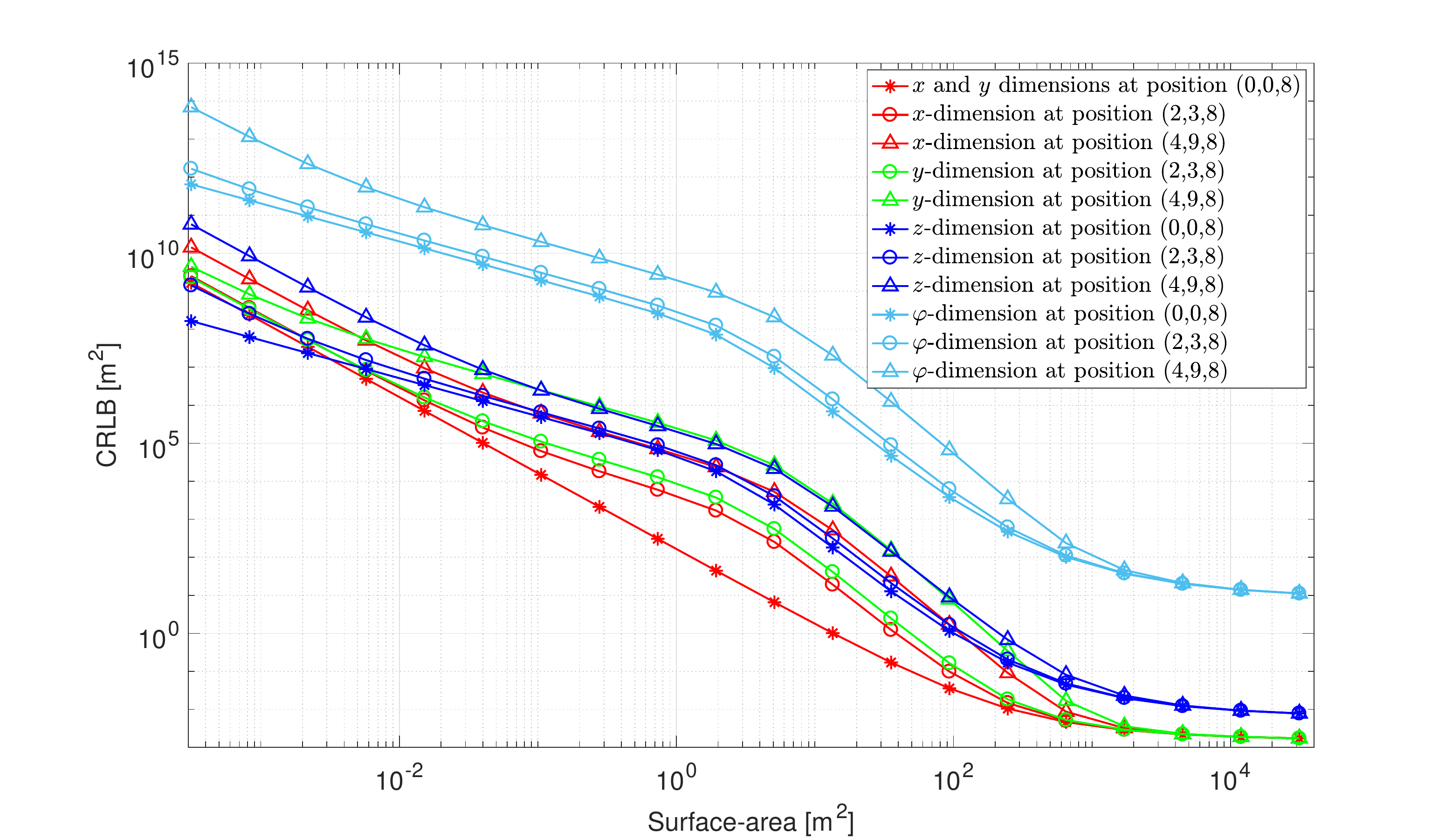}}
\vspace*{-8mm}
\caption{\label{fig10}The CRLB evaluated with unknown phase $\varphi$ for a terminal with different locations, both on and away from the CPL.}
\vspace*{-12mm}
\end{center}
\end{figure}

\subsection{CRLB with an Unknown Phase $\varphi$}
Next, we evaluate the CRLB for positioning with an unknown phase $\varphi$ presented as discussed in Sec. V. As can be seen in Fig. \ref{fig10}, when the terminal is away from the CPL, the CRLBs for all dimensions are increased and the curves have similar shapes. For all three Cartesian dimensions, the CRLB starts to decrease in the third-order of the surface-area when $R$ is larger than a certain threshold as explained in (\ref{Czpn3}). More interestingly, the CRLBs for $x$ and $y$ dimensions are lower than that for $z$-dimension when there is an unknown phase $\varphi$ present in the signal model. Furthermore, the behaviors of CRLB for a terminal not on the CPL is slightly different from the case located on the CPL. As can be seen, when $R$ is small, the CRLB decreases first quadratically in the surface-area instead of linearly, which is mainly because that the CRLB converges to the case with known $\varphi$, since the CRLB is so large that the impact of an unknown $\varphi$ is negligible. The CRLB for phase $\varphi$ is much higher than for the other Cartesian dimensions, and is around $\frac{4\pi^2}{\lambda^2}$ times of the CRLB for $z$-dimension as shown in (\ref{CzCwpn}), which basically means that the estimation of $\varphi$ is highly inaccurate unless at a very high SNR.

\subsection{CRLB with Centralized and Distributed Deployments of the LIS}
Finally, we evaluate the CRLB with the centralized and distributed deployments as discussed in Sec. VI. We set $W\!=\!H\!=\!4$ and $z_0\!=\!8$. All curves are obtained with numerical integrations without any approximations. We compare the CRLB with different deployments depicted in Fig. \ref{fig4}, that is, a single LIS, 4 small LISs, and 16 smaller LISs, with the same total surface-area. 

As shown in Fig. \ref{fig11} for a terminal on the CPL, when (\ref{thresh2}) is fulfilled, i.e., $R\!\leq\!\sqrt{\frac{W^2+H^2}{6}}\!=\!2.31$, the distributed deployments with 4 and 16 small LISs render lower CRLBs than the centralized deployment for $x$ and $y$ dimensions, while the CRLB for $z$-dimension remains the same. When $R$ increases beyond the threshold, the distributed deployments become worse for $x$ and $y$ dimensions, although the CRLB for $z$-dimension is slightly better.

\begin{figure}[t]
\begin{center}
\vspace*{-6mm}
\hspace*{-4mm}
\scalebox{0.42}{\includegraphics{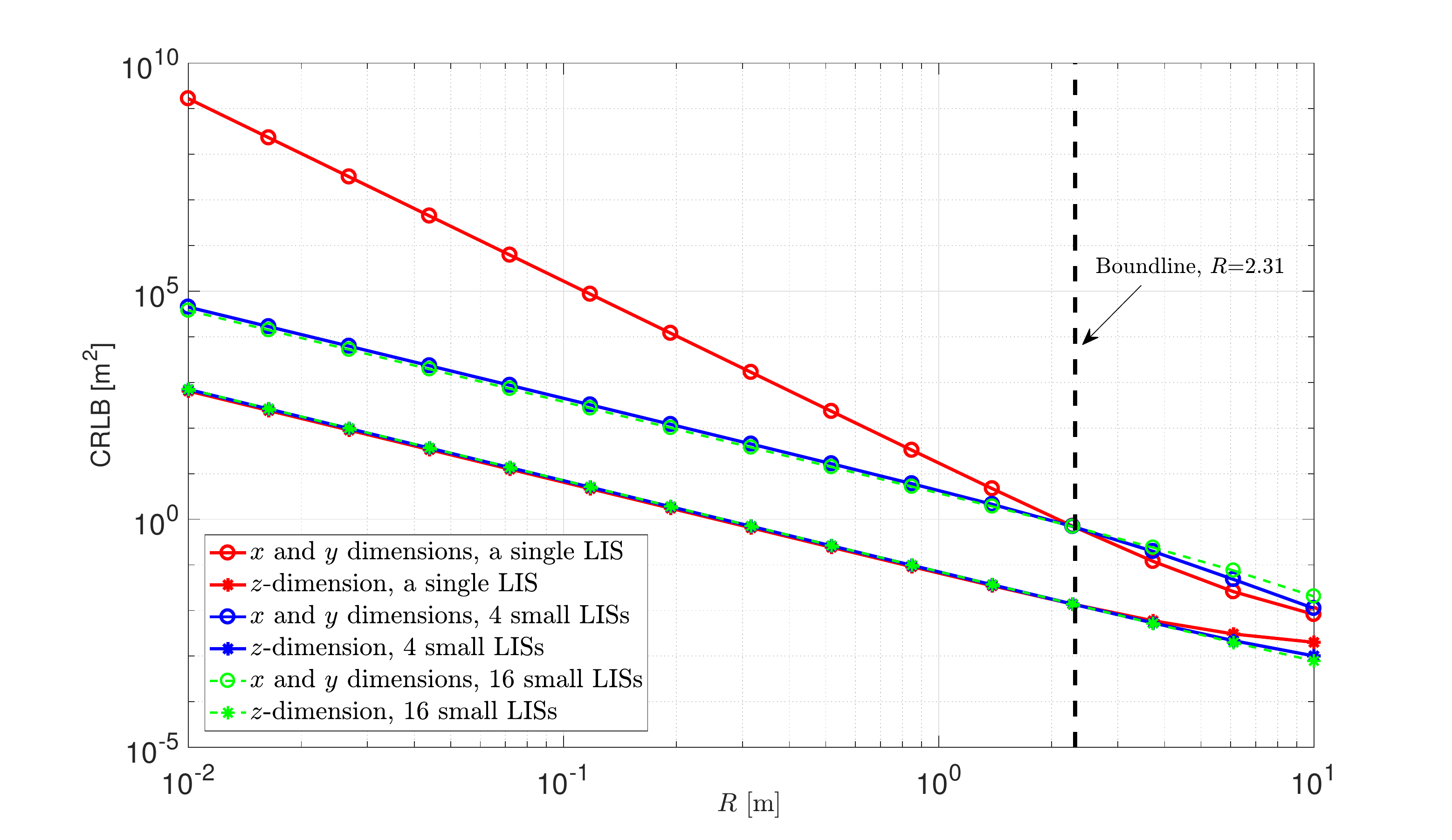}}
\vspace*{-8mm}
\caption{\label{fig11}The CRLB with different deployments of the LIS for a terminal on the CPL with $z_0\!=\!8$ and different radius $R$.}
\vspace*{-12mm}
\end{center}
\end{figure}

\begin{figure}[t]
\begin{center}
\vspace*{-6mm}
\hspace*{-4mm}
\scalebox{0.42}{\includegraphics{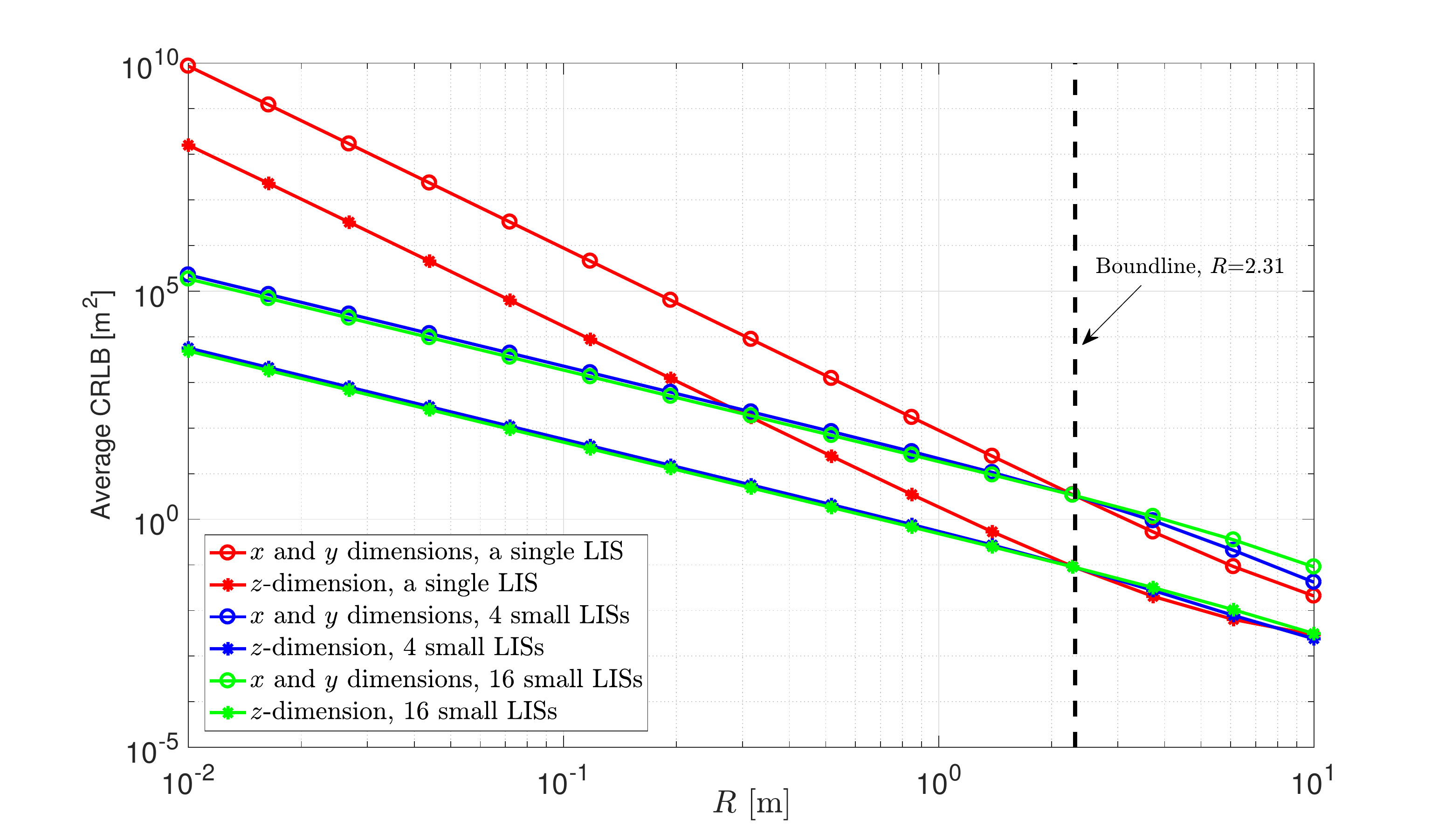}}
\vspace*{-8mm}
\caption{\label{fig12}The average CRLB with different deployments of the LIS for 1000 uniformly distributed terminal locations, and with $z_0\!=\!8$ and different radius $R$.}
\vspace*{-6mm}
\end{center}
\end{figure}

\begin{figure}
\begin{center}
\vspace*{-4mm}
\hspace*{-4mm}
\scalebox{0.42}{\includegraphics{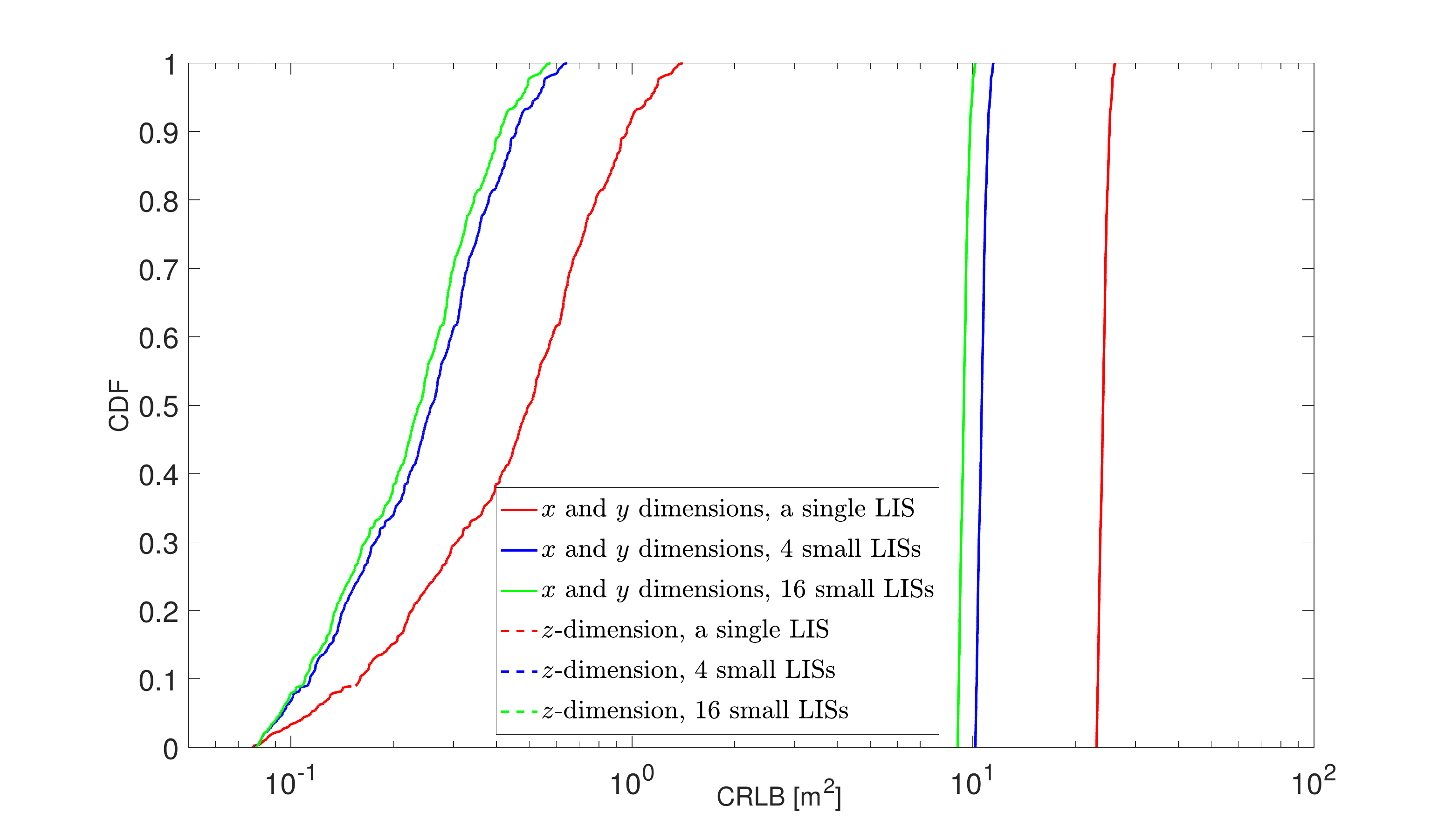}}
\vspace*{-8mm}
\caption{\label{fig13}The CDF of CRLB with different deployments of the LIS for 1000 uniformly distributed terminals with $R\!=\!1.39$.}
\vspace*{-12mm}
\end{center}
\end{figure}

In order to evaluate the average positioning performance, we draw 1000 terminals with coordinates of $x$ and $y$ dimensions uniformly distributed in [-2, 2], and $z_0\!=\!12$ for all terminals. In Fig. \ref{fig12} we plot the average CRLB for different dimensions. As can be seen, the average CRLB for all three dimensions are significantly improved with the distributed deployments. The average CRLB with 4 small LISs, each small LIS has a radius 0.005, can achieve the same average CRLB for a single LIS with $R\!=\!0.2$, that is, the surface-area needed for the distributed deployment is only 0.25\% of that for a centralized deployment when $R$ is small. As $R$ increases, different deployments converge to each other as expected. Further splitting the 4 small LISs into 16 smaller LISs provides marginal gains, but a likely cost of more stringent hardware requirements to achieve phase calibration and cooperation among the small LISs.

The cumulative distribution functions (CDF) of the CRLB are plotted in Fig. \ref{fig13}, where we can see that the CRLBs for all three Cartesian dimensions with a distributed deployment comprising 4 small LISs are significantly improved compared to a single centralized LIS. The CRLBs for $x$ and $y$ dimensions are relatively larger than that for $z$-dimension, however, the values of CRLB are also more concentrated than those for $z$-dimension. With 4 small LISs, the values of CRLB also become concentrated, which means that the overall positioning performance is improved with a distributed deployment of the LIS.

\section{Summary}
In this paper, we have derived the Fisher-information and Cram\'{e}r-Rao lower bounds (CRLBs) for positioning with large intelligent surfaces (LIS). For a terminal on the central perpendicular line (CPL), the CRLBs are derived in closed-form. For other positions we alternatively provide approximations in closed-form to compute the Fisher-information and CRLB which are shown to be accurate. We have also shown that, under mild conditions the CRLBs for $x$ and $y$ dimensions decrease quadratically in the surface-area of the deployed LIS. For $z$-dimension, the CRLB decreases linearly in the surface-area for a terminal on the CPL. When the terminal is away from the CPL, the CRLBs for all Cartesian dimensions increase dramatically and decrease quadratically in the surface-area of the LIS. 

Further, we have also analyzed the CRLBs for positioning in the presence of a random unknown phase $\varphi$ in the received signal model. We have shown that, the CRLBs are dramatically increased by the unknown phase, and in general the CRLBs for all dimensions decrease in the third-order of the surface-area, provided that the surface-area exceeds a certain threshold. We have also shown that, for an infinitely large LIS, the CRLB for $z$-dimension with an unknown phase is 6 dB higher than that with a known $\varphi$, and the CRLB for estimating $\varphi$ converges to a constant independent of the wavelength $\lambda$. 

Furthermore, we compare centralized and distributed deployments of the LIS and show that, the distributed deployments have the potential to extend the coverage of terminal-positioning and can provide better average CRLBs for all dimensions.

\section*{Appendix A: Proof of Theorem 1}
For a terminal on the CPL, we have $x_0\!=\!y_0\!=\!0$, and then the first-order derivatives with respect to $x$ and $y$ are equal to
{\setlength\arraycolsep{0pt}  \bea \label{dev11} \Delta s_1&=&\frac{\sqrt{z_0}x}{2\sqrt{\pi}}\!\left(\!
\frac{3}{2}\eta^{-\frac{7}{4}} \!+\! \frac{2\pi j }{\lambda}{\eta}^{-\frac{5}{4}}\!\right)\!\exp\left(\!-\frac{2\pi j\sqrt{\eta}}{\lambda} \right)\!,   \\
\label{dev22} \Delta s_2&=&\frac{\sqrt{z_0}y}{2\sqrt{\pi}}\left(
\!\frac{3}{2}\eta^{-\frac{7}{4}} \!+\! \frac{2\pi j }{\lambda}{\eta}^{-\frac{5}{4}}\!\right)\!\exp\left(\!-\frac{2\pi j\sqrt{\eta}}{\lambda}\right)\!,   \eea}
\hspace{-1.4mm}where $\eta=z_0^2+y^2+x^2$, and the first-order derivative with respect to $z$ is in (\ref{dev3}). Since $\eta$ is an even function with respect to $x$ and $y$, the cross-terms of different dimensions in the Fisher-information matrix vanish, and we obtain a diagonal Fisher-information matrix with diagonal elements being
\bea \label{appA1} I_{ii}=\iint_{x^2+y^2\leq R^2} |\Delta s_i|^2\mathrm{d}x\mathrm{d}y. \eea
Calculating (\ref{appA1}) directly yields
{\setlength\arraycolsep{2pt}\bea  \label{appA2} I_{11}&=&I_{22}=\frac{z_0}{4\pi}\left(\frac{9}{4}g_1(7)+ \frac{4\pi^2}{\lambda^2}g_1(5)\right),  \\
\label{appA3} I_{33}&=&\frac{z_0^3}{4\pi}\left(\frac{1}{4z_0^4}g_3(3)+\left(\frac{4\pi^2}{\lambda^2}- \frac{3}{2z_0^2}\right)g_3(5)+ \frac{9}{4}g_3(7)\right), \eea}
\hspace{-1.4mm}Utilizing the results in (\ref{g1n}) and (\ref{g2n}) and after some manipulations, the Fisher-information for different dimensions are then in (\ref{Ixy}) and (\ref{Iz}).

\section*{Appendix B: Proof of Property 2}
With conditions (\ref{cond1})-(\ref{cond2}),  it holds that
 \bea \label{eta1} \eta=z_0^2+(y-y_0)^2+(x-x_0)^2=z_1^2+x^2+y^2+o(z_1),  \eea
and the noiseless signal (\ref{md1}) can be written as
 \bea  \label{md3} \hat{s}_{x_0, y_0, z_0}(x,y)=\sqrt{\frac{z_0}{z_1}}s_{0, 0, z_1}(x,y)\left(1+\mathop{o}(1)\right). \eea
The parameters $\alpha$ and $\beta$ are defined in (\ref{alpha}) and (\ref{beta}). Note that with conditions in (\ref{cond1})-(\ref{cond2}), the derivatives in (\ref{dev1})-(\ref{dev3}) can be approximated as
{\setlength\arraycolsep{2pt}\bea  \Delta s_1&=& \frac{\sqrt{\pi z_0}\left(x -x_0\right) j }{\lambda}{\eta}^{-\frac{5}{4}}\!\exp\left(\!-\frac{2\pi j}{\lambda} \sqrt{\eta}\right)\left(1+\mathcal{O}\left(\frac{\lambda}{z_1}\right)\right)\!,\\
 \Delta s_2&=&\frac{\sqrt{\pi z_0}\left(y-y_0\right)  }{\lambda}{\eta}^{-\frac{5}{4}}\!\exp\left(\!-\frac{2\pi j}{\lambda} \sqrt{\eta}\right)\left(1+\mathcal{O}\left(\frac{\lambda}{z_1}\right)\right)\!,\\
 \Delta s_3&=&-\frac{\sqrt{\pi}z_0^{\frac{3}{2}} j }{\lambda}{\eta}^{\frac{5}{4}}\!\exp\left(\!-\frac{2\pi j}{\lambda} \sqrt{\eta}\right)\left(1+\mathcal{O}\left(\frac{\lambda}{z_1}\right)\right)\!.\eea}
\hspace{-1.4mm}First we consider the Fisher-information for the $z$-dimension. Since
\bea \label{dev44} \Delta s_3\left(\Delta s_3\right)^{\!\ast}=\frac{\pi z_0^3 }{\lambda^2 }{\eta}^{-\frac{5}{2}}\left(1+\mathcal{O}\left(\frac{\lambda^2}{z_1^2}\right)\right)\!,\eea
and utilizing (\ref{md3}), the Fisher-information for the $z$-dimension with coordinates $(x_0, y_0, z_0)$ can be approximated based on that with coordinates $(0, 0, z_0)$, which is
\bea  \label{I33apendB} I_{33}= \beta\left(1+o\left(\frac{\lambda}{z_1}\right)\right)\!.\eea
Similarly, as
\bea \label{dev33}  \Delta s_1\left( \Delta s_3\right)^{\!\ast}&=&\left(-\frac{\pi x z_0^2  }{\lambda^2}{\eta}^{-\frac{5}{2}}+\frac{\pi x_0 z_0^2  }{\lambda^2}{\eta}^{-\frac{5}{2}}\right)\left(1+\mathcal{O}\left(\frac{\lambda^2}{z_1^2}\right)\right)\!, \eea
and the integrals over the term $-\frac{\pi x z^2}{\lambda^2} {\eta}^{-\frac{5}{2}}$ (with respect to $x$ and $y$ in $\eta$) in $ \Delta s_1\left( \Delta s_3\right)^{\!\ast}$ is zero since it is an odd function in $x$, then by directly comparing the remaining term in (\ref{dev33}) to (\ref{dev44}), after integration over $x$ and $y$, it holds that
\bea \label{I13} I_{13}=\frac{x_0}{z_0}I_{33}\left(1+o\left(\frac{\lambda}{z_1}\right)\right)\!,\eea
Furthermore, it also hold that
\bea \label{appA2}  \Delta s_1\left( \Delta s_1\right)^{\!\ast}=\frac{\pi z_0(x^2-2xx_0+x_0^2)  }{\lambda^2 }{\eta}^{-\frac{5}{2}}\left(1+\mathcal{O}\left(\frac{\lambda^2}{z_1^2}\right)\right)\!.\eea
Using (\ref{md3}), the integrals of the first term $\frac{\pi z_0x^2 }{\lambda^2 }{\eta}^{-\frac{5}{2}}$ in (\ref{appA2}) can be approximated by $\alpha$, which is calculated based on the Fisher-information of the $x$-dimension with coordinates $(0, 0,z_1)$. Then, the second term $\frac{xx_0 }{\lambda^2 }{\eta}^{-\frac{5}{2}}$ is an odd function in $x$ and the integral of it is zero. At last, comparing the last term $\frac{\pi z_0x_0^2 }{\lambda^2 }{\eta}^{-\frac{5}{2}}$ to (\ref{dev44}) yields
\bea  I_{11}= \left(\alpha+\frac{x_0^2}{z_0^2}\beta\right)\left(1+o\left(\frac{\lambda}{z_1}\right)\right)\!.\eea

Using the symmetry between $x$ and $y$ dimensions, it can also be shown that
{\setlength\arraycolsep{2pt}\bea  I_{23}&=&\frac{y_0}{z_0} \beta\left(1+o\left(\frac{\lambda}{z_1}\right)\right)\!,  \\
 \label{I22appA} I_{22}&=&  \left(\alpha+\frac{y_0^2}{z_0^2}\beta\right)\left(1+o\left(\frac{\lambda}{z_1}\right)\right)\!.\eea}
\hspace{-1.4mm}Finally, as 
\bea \label{appA3}  \Delta s_1\left( \Delta s_2\right)^{\!\ast}=\frac{\pi z_0(xy-xx_0-yy_0+x_0y_0)  }{\lambda^2 }{\eta}^{-\frac{5}{2}}\left(1+\mathcal{O}\left(\frac{\lambda^2}{z_1^2}\right)\right)\!,\eea
the integrals of the first three terms in (\ref{appA3}) vanish as they are odd functions either in $x$ or $y$. Comparing the last term to (\ref{dev44}), it holds that 
\bea  \label{I12} I_{12}= \frac{x_0y_0}{z_0^2}\beta\left(1+o\left(\frac{\lambda}{z_1}\right)\right)\!.\eea
Noting that the Fisher-information matrix is symmetric, and combining (\ref{I13})-(\ref{I12}), the Fisher-information matrix is in (\ref{Imat}), which completes the proof.

\section*{Appendix C: Proof of Theorem 2}
First of all, since the unknown phase $\varphi$ only appears in the exponential terms of the first-order derivatives, it does not appear in the Fisher-information matrix, and the Fisher-information for $x$, $y$, and $z$ dimensions remain the same. 

Next, we compute the cross-terms between $\varphi$-dimension and the other dimensions. Following the similar arguments as in the proof of Property 2, it can be shown that,
{\setlength\arraycolsep{2pt} \bea \label{appD1}  I_{14}&=&\iint_{x^2+y^2\leq R}\Re\!\left\{\Delta s_1 \left(\Delta s_4\right)^{\ast}\right\}\mathrm{d}x\mathrm{d}y =\frac{x_0z_0}{2\lambda}g_3(4),  \\
\label{appD2} I_{24}&=&\iint_{x^2+y^2\leq R}\Re\!\left\{\Delta s_1 \left(\Delta s_4\right)^{\ast}\right\}\mathrm{d}x\mathrm{d}y =\frac{y_0z_0}{2\lambda}g_3(4),  \\
\label{appD3}  I_{34}&=&\iint_{x^2+y^2\leq R}\Re\!\left\{\Delta s_3 \left(\Delta s_4\right)^{\ast}\right\}\mathrm{d}x\mathrm{d}y =\frac{z_0^2}{2\lambda}g_3(4), \eea}
\hspace{-1.4mm}and
\bea \label{appD4}  I_{44}=\iint_{x^2+y^2\leq R}|\Delta s_4 |^2\mathrm{d}x\mathrm{d}y = \frac{z_0}{4\pi}g_3(3), \eea
where $g_3(n)$ is defined in (\ref{g3}). 

\section*{Appendix D: Proof of Property 4}
Inserting the expressions of $I_{33}$, $I_{34}$, $I_{44}$ given in (\ref{I33apendB}), (\ref{appD3}), (\ref{appD4}) to (\ref{Czpn}) and (\ref{Cwpn}), we have the CRLBs for $z$-dimension and phase $\varphi$ equal to
\bea\label{Czpn3}  C_z&=&\left(\frac{z_0^3}{4\pi}\left(\frac{1}{4z_0^4}g_3(3)+\left(\frac{4\pi^2}{\lambda^2}- \frac{3}{2z_0^2}\right)g_3(5)+ \frac{9}{4}g_3(7)\right) -\frac{\pi z_0^3}{\lambda^2}\frac{g_3^2(4)}{g_3(3)}\right)^{-1} ,\eea
and
\bea \label{Cwpn3} C_\varphi =\left(\frac{z_0}{4\pi}g_3(3) -\frac{\pi z_0}{\lambda^2}\frac{g_3^2(4)}{\left(\frac{1}{4z_0^4}g_3(3)+\left(\frac{4\pi^2}{\lambda^2}- \frac{3}{2z_0^2}\right)g_3(5)+ \frac{9}{4}g_3(7)\right)}\right)^{-1},\eea
respectively. For a terminal on the CPL, using the formula of $g_3(n)$ in (\ref{g2n}) yields
{\setlength\arraycolsep{2pt} \bea \label{g23}  g_3(3)&=&\frac{1}{z_0}-\frac{1}{R^2+z_0^2}, \\
\label{g24}g_3(4)&=&\frac{1}{2}\left(\frac{1}{z_0^2}-\frac{1}{\sqrt{R^2+z_0^2}} \right)\!, \\
\label{g25}g_3(5)&=&\frac{1}{3}\left(\frac{1}{z_0^3}-\frac{1}{(R^2+z_0^2)^{\frac{3}{2}}} \right)\!, \\
\label{g27} g_3(7)&=&\frac{1}{5}\left(\frac{1}{z_0^5}-\frac{1}{(R^2+z_0^2)^{\frac{5}{2}}} \right)\!. \eea}
\hspace{-1.4mm}Inserting (\ref{g23})-(\ref{g27}) back into (\ref{Czpn3}) and (\ref{Cwpn3}), after some manipulations, the CRLBs for $z$-dimension and phase $\varphi$ are in (\ref{Czpn1}) and (\ref{Cwpn1}), respectively.

\bibliographystyle{IEEEtran}

\begin{thebibliography}{99}

\bibitem{HRE172} S. Hu, F. Rusek, and O. Edfors, \lq\lq{}Cram\'er-Rao Lower Bounds for Positioning with Large Intelligent Surfaces,\rq\rq{} \textit{accept in} IEEE Vehicular Technology Conference (VTC), Fall, 2017, \textit{available at: https://arxiv.org/abs/1702.03131}.

\bibitem{HRE171} S. Hu, F. Rusek, and O. Edfors, \lq\lq{}The potential of using large antenna arrays on
intelligent surfaces,\rq\rq{} \textit{accepted in} IEEE Vehicular Technology Conference (VTC), available at: \textit{https://arxiv.org/abs/1702.03128}, Spring, 2017.

\bibitem{ML17} X. Huang, \lq\lq{}Machine learning and intelligent communications,\rq\rq{} Proc. of International Conference on Machine Learning and Intelligent Communications (MLICOM), Shanghai, China, Aug. 27-28, 2016.

\bibitem{M10} T. L. Marzetta, \lq\lq{}Noncooperative cellular wireless with unlimited numbers of base station antennas,\rq\rq{} \textit{IEEE Trans. on Wireless Commun.}, vol. 9, no. 11, pp. 3590-3600, Nov. 2010.

\bibitem{MM12} F. Rusek, D. Persson, B. K. Lau, E. G. Larsson, T. L. Marzetta, O. Edfors, and F. Tufvesson, \lq\lq{}Scaling up MIMO: Opportunities and challenges with very large arrays,\rq\rq{} \textit{IEEE Signal Process. Magazine} vol. 30, no. 1, pp. 40-60, Dec. 2012.

\bibitem{MM14} E. G. Larsson, F. Tufvesson, O. Edfors, and T. L. Marzetta, \lq\lq{}Massive MIMO for next generation wireless systems,\rq\rq{}  \textit{IEEE Commun. Mag.,} vol. 52, no. 2, pp. 186-195, Feb. 2014.

\bibitem{AZ14} J. Andrews, S. Buzzi, W. Choi, S. V. Hanly, A. Lozano, A. Soong, and J. Zhang, \lq\lq{}What will 5G be?\rq\rq{}, \textit{IEEE Journal on Selected Areas in Commu.,} vol. 32, no. 6, pp. 1065-1082, Jun. 2014.


\bibitem{IOT} L. Atzori, A.  Iera, and G. Morabito, \lq\lq{}The internet of things: A survey\rq\rq{}. \textit{Computer networks, Elsevier}, vol. 54, no. 15, pp. 2787-2805, Oct. 2010.

 \bibitem{ewall} A. Puglielli, N. Narevsky, P. Lu, T. Courtade, G. Wright, B. Nikolic, and E. Alon, \lq\lq{}A scalable massive MIMO array architecture based on common modules, \rq\rq{} \textit{In Proc.} IEEE International Conference on Communications (ICC), workshop on 5G and beyond, May 2015.

\bibitem{JY07} S. Al-Jazzar, J. Caffery, and H. R. You, \lq\lq{}Scattering-model-based methods for TOA location in NLOS environments\rq\rq{}, \textit{IEEE Trans. on Vehicular Technology}, wol. 56, no. 2. pp. 583-593, Mar. 2007.
	
\bibitem{WW15} S. Wu, D. Xu, and H. Wang, \lq\lq{}Adaptive NLOS mitigation location algorithm in wireless cellular network,\rq\rq{} \textit{Wireless Personal Commu.: An International Journal,} vol. 84, no. 4, pp. 3143-3156, Oct. 2015.

\bibitem{MT11} R. Mautz  and S. Tilch, \lq\lq{}Survey of optical indoor positioning systems,\rq\rq{} International Conference on Indoor Positioning and Indoor Navigation (IPIN), Nov. 2011. 

\bibitem{F10} A. F. Molisch, \textit{Wireless communications}, the second edition, Wiley-IEEE Press, Nov. 2010.

\bibitem{K93} S. M. Kay, \lq\lq{}Fundamentals of statistical signal processing, volume {I}: Estimation theory,\rq\rq{} Prentice Hall signal processing series, 1993.

\bibitem{R87} W. Rudin, \textit{Real and complex analysis}, the third edition, New York, McGraw-Hill Book Co., 1987.

\bibitem{JF17} J. Vieira, F. Rusek, O. Edfors, S. Malkowsky, L. Liu, and F. Tufvesson, \lq\lq{}Reciprocity calibration for massive MIMO: Proposal, modeling and validation\rq\rq{}, accepted for publication in \textit{IEEE Trans. on Wireless Commu.}, available at: \textit{https://arxiv.org/abs/1606.05156}, Jun. 2016.


\end{thebibliography}

\end{document}